\documentclass[article]{article}
\usepackage{color}
\usepackage{graphicx}
\usepackage{amsmath}
\usepackage{amssymb}
\usepackage{mathrsfs}
\usepackage[numbers,sort&compress]{natbib}
\usepackage{amsthm}
\numberwithin{equation}{section}
\usepackage{pgf}
\usepackage{CJK}
\usepackage{graphicx}
\usepackage{tikz}
\usepackage{stmaryrd}
\usepackage{graphicx}
\usepackage{hyperref}
\usepackage[latin1]{inputenc}
\usepackage[T1]{fontenc}
\usepackage[english]{babel}
\usepackage{listings}
\usepackage{xcolor,mathrsfs,url}
\usepackage{amssymb}
\usepackage{amsmath}
\usepackage{ifthen}
\newtheorem{theorem}{Theorem}
\newtheorem{lemma}{Lemma}
\newtheorem{proposition}{Proposition}

\usepackage {caption}
\usepackage{float}
\renewcommand{\Re}{\operatorname{Re}}
\renewcommand{\Im}{\operatorname{Im}}
\usepackage{subfigure}

\begin{document}

\title{ Inverse scattering transform and multiple high-order pole solutions  for the Gerdjikov-Ivanov equation   under the zero/nonzero background }
\author{Zechuan ZHANG$^1$\thanks{\ Email address: 17110180013@fudan.edu.cn } and Engui FAN$^{1}$\thanks{\ Corresponding author and email address: faneg@fudan.edu.cn } }
\footnotetext[1]{ \  School of Mathematical Sciences  and Key Laboratory of Mathematics for Nonlinear Science, Fudan University, Shanghai 200433, P.R. China.}

\date{ }

\maketitle
\begin{abstract}
\baselineskip=17pt

In this article, the inverse scattering transform is considered for the Gerdjikov-Ivanov equation with zero and non-zero boundary conditions
  by a  matrix Riemann-Hilbert (RH)   method.
 The formula of the soliton solutions are established by Laurent expansion to the RH problem.
The method we used is different from computing solution with simple poles since the residue conditions here are hard to obtained.
The formula of multiple soliton solutions with one high-order pole and $N$ multiple high-order poles are obtained respectively.
The dynamical properties and   characteristic  for the  high-order pole solutions are further analyzed.



\end{abstract}

\newpage

 \tableofcontents

\baselineskip=17pt

\section {Introduction}

\quad

The inverse scattering transform method plays a significant role during the discovery process of the exact solutions of completely integrable systems \cite{RN19,RN16}.
As a new version of inverse scattering transform method,  the Riemann-Hilbert (RH)  approach   has become the preferred research technique to the researchers in   investigating the soliton solutions and the long-time asymptotics of integrable systems in recent years \cite{RN20,RN15}.
More recently,  the RH approach has  been widely used  to investigate the integrable systems with nonzero boundary  \cite{RN17,RN18,RN21,RN22,RN23,RN24}.
  The high-order   soliton solutions which have the same velocity and locate at the same position also have been studied   \cite{RN25,RN26,RN27}.
  The general method to obtain  the high-order solitons with the classical  inverse scattering transform (IST)  method involves
   some complicated   calculation,  especially for the case of multiple high-order poles \cite{RN28,RN29,RN30}.
In this case,   it is effective  to construct  high-order pole  soliton solutions  of integrable systems by Laurent expansion to the RH problem \cite{RN30,zy2019,zys2020}.

 It is well-known that the nonlinear Schr\"{o}dinger (NLS) equation \cite{RN1,RN2}
\begin{equation}
iq_t+q_{xx}+2|q|^2q=0  \label{NLS}
\end{equation}
  is one of the most important integrable systems, which  plays an important role and has applications in a wide variety of fields. Besides the NLS equation (\ref{NLS}), derivative NLS (DNLS) equations were also introduced to investigate the effects of high-order perturbations \cite{RN3,RN5,EN}. Among them, there are three  derivative NLS equations \cite{EN},  the first one is Kaup-Newell   equation  \cite{RN8}
 \begin{equation}
 iq_t+q_{xx}+i(|q|^2)_x=0.
 \end{equation}
The second type is the Chen-Lee-Liu equation   \cite{RN6}
 \begin{equation}
 iq_t+q_{xx}+i|q|^2q_x=0.
 \end{equation}
The third type  is  the Gerdjikov-Ivanov (GI) equation   which    takes the form   \cite{RN7}
\begin{equation}
iq_t+q_{xx}-iq^2q_x^\ast+\dfrac{1}{2}q^3q^{\ast2}=0,\label{GI}
\end{equation}
where the asterisk $*$ means the complex conjugation.   The  DNLS equations are regarded as models in a wide variety of fields such as weakly nonlinear dispersive water waves, nonlinear optical fibers, quantum field theory and plasmas \cite{APPLI1,APPLI2,APPLI4,APPLI5}.
   In plasma physics, the GI equation (\ref{GI})  is a model for Alfv$\acute{e}$n waves propagating parallel to the ambient magnetic field, where $q$ being the transverse magnetic field perturbation and $x$ and $t$ being space and time coordinates, respectively  \cite{APPLI6,APPLI7}.
 The GI equation has been studied through many methods. For instance, the Darboux transformation \cite{RN9},
the nonlinearization \cite{RN11,RN12}, the similarity reduction, the bifurcation theory and others   \cite{RN13,RN14}.
Especially,  RH method is used to construct  N-soliton of the   GI  equation  with zero boundary \cite{NZG}. Recently, we  used the RH method  to construct simple
pole solutions  of the   GI  equation
with  nonzero boundary conditions \cite{RN31}.

In this article, we further investigate the inverse scattering transform and    high-order   solutions  of GI equation (\ref{GI}) with zero boundary condition
\begin{equation}
q(x,t)\rightarrow 0,\hspace{0.5cm}x\rightarrow \pm\infty,\label{bc0}
\end{equation}
and the following nonzero boundary conditions
\begin{equation}
q(x,t) \sim  q_{\pm}e^{-\frac{3}{2}iq_0^4t+iq_0^2x},\hspace{0.5cm}x\to{\pm}\infty,\label{bc1}
\end{equation}
where $ \left|q_{\pm}\right|= q_0 > 0$, and $q_{\pm}$ are independent of $x, t$.
 The formula of multiple soliton  solutions of the GI equation are obtained, which correspond to multiple high-order poles of the RH problem.

This paper is organized as follows. In section 2, we construct the RH problem of GI equation (\ref{GI}) with zero boundary condition, and display the relationship between the solutions of the RH problem and GI equation, then we derive the formula of    single high-order   solutions  and multiple high-order   solutions  of  GI equation. In section 3, by the same method, we give the one single high-order soliton solution and multiple high-order soliton solutions of  GI equation (\ref{GI}) with nonzero boundary condition. We give the patterns for both zero and nonzero boundary conditions.

\section{IST with zero boundary and high-order pole}

\subsection{Spectral analysis}

 \subsubsection{ Eigenfunctions  and scattering matrix  }

It is well-known that the GI equation (\ref{GI}) admits the Lax pair \cite{RN9}
\begin{equation}
\psi_x = X\psi,\hspace{0.5cm}\psi_t = T\psi, \label{lax1}
\end{equation}
where
\begin{equation}
X=-ik^2\sigma_3+kQ-\frac{i}{2}Q^2\sigma_3,
\end{equation}
\begin{equation}
T=-2ik^4\sigma_3+2k^3Q-ik^2Q^2\sigma_3-ikQ_x\sigma_3+\frac{1}{2}(Q_xQ-QQ_x)+\frac{i}{4}Q^4\sigma_3,
\end{equation}
and
\begin{equation}
\sigma_3=\left(\begin{array}{cc}
                      1 & 0   \\
                      0 & -1
                      \end{array}\right),\hspace{0.5cm}Q=\left(\begin{array}{cc}
                                                                0 & q  \\
                                                                -q^{\ast} &0
                                                               \end{array}\right).
\end{equation}

With zero boundary (\ref{bc0}), asymptotic spectra problem of the Lax pair (\ref{lax1}) becomes
\begin{equation}
\psi_x=X_\pm\psi,\hspace{0.5cm}\psi_t=T_\pm\psi,
\end{equation}
where
\begin{equation}
X_\pm=-ik^2\sigma_3,\hspace{0.5cm}T_\pm=-2ik^3\sigma_3.
\end{equation}
We define the Jost eigenfunctions $\phi_\pm(x,t,k)$ as the simultaneous solutions of both parts of the Lax pair, so that
\begin{equation}
\phi(x,t,k)=\psi(x,t,k)e^{i\theta(k)\sigma_3},
\end{equation}
where $\theta(k)=k^2(x+2k^2t)$, then
\begin{equation}
\phi_\pm(x,t,k)\rightarrow I,\hspace{0.5cm}x\rightarrow \pm\infty.
\end{equation}
Meanwhile, $\phi_\pm$ acquire the equivalent Lax pair
\begin{subequations}
\begin{equation}
\phi_x(x,t,k)+ik^2[\sigma_3,\phi(x,t,k)]=\Delta X_\pm\phi(x,t,k);
\end{equation}
\begin{equation}
\phi_t(x,t,k)+2ik^4[\sigma_3,\phi(x,t,k)]=\Delta T_\pm\phi(x,t,k),
\end{equation}
\end{subequations}
where $\Delta X_\pm=X-X_\pm$ and $\Delta T_\pm=T-T_\pm$.

Since $\psi_\pm(x,t,k)$ are two fundamental matrix solutions, there exists a constant matrix $S(k)$   such that
\begin{equation}
\psi_+(x,t,k)=\psi_-(x,t,k)S(k), \label{scattering0}
\end{equation}
where $S(k)=(s_{ij}(k))_{2\times2}$ is referred to the scattering matrix and its entries as the scattering coefficients. It follows from (\ref{scattering0}) that $s_{ij}$ has the Wronskian representation:
\begin{subequations}
\begin{equation}
s_{11}(k)=\text{Wr}(\psi_{+,1},\psi_{-,2}),\hspace{0.5cm}s_{12}(k)=\text{Wr}(\psi_{+,2},\psi_{-,2}),
\end{equation}
\begin{equation}
s_{21}(k)=\text{Wr}(\psi_{-,1},\psi_{+,1}),\hspace{0.5cm}s_{22}(k)=\text{Wr}(\psi_{-,1},\psi_{+,2}).
\end{equation}
\end{subequations}

\subsubsection{ Analyticity  }

As a result, the Volterra integral equations are
\begin{equation}
\phi_\pm(x,t,k)=I+\int_{\pm\infty}^x e^{-ik^2(x-y)\sigma_3}\Delta X_\pm\phi_\pm(y,t,k)e^{ik^2(x-y)\sigma_3}dy.
\end{equation}
We define $D^+$, $D^-$ and $\Sigma$ as
\begin{equation}
D^+:=\{k\in\mathbb{C}:\text{Re}k\text{Im}k>0\},\hspace{0.2cm}D^-:=\{k\in\mathbb{C}:\text{Re}k\text{Im}k<0\},\hspace{0.2cm}\Sigma:=\mathbb{R}\cup i\mathbb{R}.
\end{equation}

\begin{proposition}
Suppose that $q(x,t)\in L^1(\mathbb{R})$ and $\phi_{\pm,j}(x,t,k)$ denotes the $j$th column of $\phi_\pm(x,t,k)$, then $\phi_\pm(x,t,k)$ have the following properties:

$\bullet$ $\phi_{-,1}$,$\phi_{+,2}$ and $s_{22}$  are analytic in $D^+$ and continuous in $D^+\cup\Sigma$.

$\bullet$ $\phi_{+,1}$,  $\phi_{-,2}$ and $s_{11}$ are analytic in $D^-$ and continuous in $D^-\cup\Sigma$.

$\bullet$  $s_{12}$ and $s_{21}$ are continuous on  $\Sigma$.

\end{proposition}

As usual, the reflection coefficients $r(k)$ are defined as
\begin{equation}
r(k)=\frac{s_{12}(k)}{s_{22}(k)},\hspace{0.5cm}\tilde{r}(k)=\frac{s_{21}(k)}{s_{11}(k)},\hspace{0.5cm}k\in\Sigma.
\end{equation}

\subsubsection{Symmetries}

\begin{proposition}
The Jost solution, scattering matrix and reflection coefficients satisfy the following reduction conditions
\begin{itemize}
\item The first symmetry reduction
\begin{equation}
\phi_\pm(x,t,k)=\sigma_2\phi_\pm^*(x,t,k^*)\sigma_2,\hspace{0.5cm}S(k)=\sigma_2S(k^*)^*\sigma_2,\hspace{0.5cm}r(k)=-\tilde{r}(k^*)^*,\label{symmetry01}
\end{equation}
\item The second symmetry reduction
\begin{equation}
\phi_\pm(x,t,k)=\sigma_1\phi_\pm^*(x,t,-k^*)\sigma_1,\hspace{0.5cm}S(k)=\sigma_1S(-k^*)^*\sigma_1,\hspace{0.5cm}r(k)=\tilde{r}(-k^*)^*,\label{symmetry02}
\end{equation}
\end{itemize}
\end{proposition}
where
\begin{equation}
\sigma_1=\left(\begin{array}{cc}
               0 & 1\\
               1 & 0
               \end{array}\right),\hspace{0.5cm}\sigma_2=\left(\begin{array}{cc}
                                                                0 & -i\\
                                                                i & 0
                                                                \end{array}\right).
\end{equation}

\subsubsection{Asymptotic behaviors}
To solve the RH problem in the next section, it is necessary to discuss the asymptotic behaviors of the modified Jost solutions and scattering matrix as $k\rightarrow\infty$ by the standard Wentzel-Kramers-Brillouin (WKB) expansions.
\begin{proposition}
The asymptotic behaviors for the modified Jost solutions and scattering matrix are given as
\begin{equation}
\phi_\pm(x,t,k)=I-\frac{i}{2k}\sigma_3Q+o(k^{-1}),\hspace{0.5cm}S(k)=I+o(k^{-1}),\hspace{0.5cm}k\rightarrow\infty
\end{equation}
\end{proposition}
Furthermore, solutions of the GI equation will be constructed by
\begin{equation}
q(x,t)=\lim_{k\rightarrow\infty}2i(k\phi_{\pm})_{12}.
\end{equation}

\subsection{Riemann-Hilbert problem}
As we all know, the equation (\ref{scattering0}) is the beginning of the formulation of the inverse  problem. We always regard it as a relation between eigenfunctions analytic in $D^+$ and those analytic in $D^-$. Thus it is necessary for us to introduce the following RH problem.
\begin{proposition}
Define the sectionally meromorphic matrix
\begin{equation}
M(x,t,k)=\Bigg\{\begin{array}{ll}
            M^-=\left(\begin{array}{cc}
           \dfrac{\phi_{+,1}}{s_{11}} & \phi_{-,2}\\
           \end{array}\right), &\text{as } k\in D^-,\\
                 M^+=\left(\begin{array}{cc}
                           \phi_{-,1} & \dfrac{\phi_{+,2}}{s_{22}}\\
                           \end{array}\right), &\text{as }k\in D^+.\\
                 \end{array}
\end{equation}
Then a multiplicative matrix RH problem is proposed:

$\bullet$ Analyticity: $M(x,t,k)$ is analytic in $\mathbb{C}\setminus\Sigma$.

$\bullet$ Jump condition
\begin{equation}
M^-(x,t,k)=M^+(x,t,k)(I-G(x,t,k)),\hspace{0.5cm}k \in \Sigma,\label{jump}
\end{equation}
where
\begin{equation}
G(x,t,k)=\left(\begin{array}{cc}
                r(k)\tilde{r}(k) & e^{2i\theta}\tilde{r}(k)\\
                -e^{-2i\theta}r(k) & 0
                \end{array}\right).
\end{equation}

$\bullet$ Asymptotic behaviors
\begin{equation}
M(x,t,k) \sim I+O(k^{-1}),\hspace{0.5cm}k \rightarrow \infty,
\end{equation}
\end{proposition}
Moreover, new solutions of the GI equation can be reconstructed by $M(x,t,k)$ as
\begin{equation}
q(x,t)=\lim_{k\rightarrow\infty}2i(kM)_{12}.\label{q0}
\end{equation}

\subsection{ Single  high-order pole solutions }

We assume $s_{22}(k)$ have high-order poles $\{\pm k_j: \Re k_j>0\}_{j=1}^N$, from the symmetries (\ref{symmetry01}) and (\ref{symmetry02}), we know that $\{\pm k_j^*:\Im k_j^*<0\}_{j=1}^N$ are high-order poles of $s_{11}(k)$. So $s_{22}(k)$ can be expanded as:
\begin{equation}
s_{22}(k)=(k^2-k_1^2)^{n_1}(k^2-k_2^2)^{n_2}\cdots (k^2-k_N)^{n_N}s_0(k),
\end{equation}
where $s_0(k)\neq0$ for all $k\in D^+$.
When $s_{22}(k)$ only has $N$ simple zeros, the RHP can be solved straightforward by the residue conditions, and the formula of $N$th order solition solutions of GI equation are obtained through (\ref{q0}). However, as $s_{22}(k)$ has multiple high-order zero points, the residue conditions are not enough, and the coefficients related to much higher negative power of $k\pm k_j$ and $k\pm k_j^*$ should be considered.
For convenience, we will consider the simplest case at first where $s_{22}(k)$ has only one higher order zero point.

Let $k_0\in D^+$ is the Nth-order pole, from the symmetries (\ref{symmetry01}) and (\ref{symmetry02}) it is obvious that $-k_0\in D^+$ also is the Nth-order pole of $s_{22}(k)$. Then $\pm k_0^*$ are the Nth-order poles of $s_{11}(k)$. The discrete spectrum is the set
\begin{equation}
\{\pm k_0,\pm k_0^*\},
\end{equation}
which can be seen in Figure 1.
\begin{center}
\begin{tikzpicture}[node distance=2cm]
\filldraw[yellow!40,line width=3] (2.8,0.01) rectangle (0.01,2.8);
\filldraw[yellow!40,line width=3] (-2.8,-0.01) rectangle (-0.01,-2.8);
\draw[->](-3,0)--(3,0)node[right]{Re$k$};
\draw[->](0,-3)--(0,3)node[above]{Im$k$};
\draw[->](0,0)--(-1.5,0);
\draw[->](-1.5,0)--(-2.8,0);
\draw[->](0,0)--(1.5,0);
\draw[->](1.5,0)--(2.8,0);
\draw[->](0,2.7)--(0,2.2);
\draw[->](0,1.6)--(0,0.8);
\draw[->](0,-2.7)--(0,-2.2);
\draw[->](0,-1.6)--(0,-0.8);
\coordinate (A) at (2.2,2.2);
\coordinate (B) at (2.2,-2.2);
\coordinate (C) at (-0.9090909090909,0.9090909090909);
\coordinate (D) at (-0.9090909090909,-0.9090909090909);
\coordinate (E) at (0.9090909090909,0.9090909090909);
\coordinate (F) at (0.9090909090909,-0.9090909090909);
\coordinate (G) at (-2.2,2.2);
\coordinate (H) at (-2.2,-2.2);
\coordinate (I) at (0,2);
\coordinate (J) at (1.414213562373095,1.414213562373095);
\coordinate (K) at (1.414213562373095,-1.414213562373095);
\coordinate (L) at (-1.414213562373095,1.414213562373095);
\coordinate (M) at (-1.414213562373095,-1.414213562373095);
\fill (J) circle (1pt) node[right] {$k_0$};
\fill (K) circle (1pt) node[right] {$k_0^*$};
\fill (L) circle (1pt) node[left] {$-k_0^*$};
\fill (M) circle (1pt) node[left] {$-k_0$};
\label{zplane2}
\end{tikzpicture}
\flushleft{\footnotesize {\bf Figure 1.}  Distribution of   the the discrete spectrum and  the contours for the RH problem on complex $k$-plane.
 }
\end{center}
Let
\begin{equation}
s_{22}(k)=(k^2-k_0^2)s_0(k),
\end{equation}
in which $s_0(k)\neq 0$ in $D^+$. According to the Laurent series expansion in poles, $r(k)$ and $r^*(k^*)$ can be respectively expanded as
\begin{subequations}
\begin{equation}
r(k)=r_0(k)+\sum_{m=1}^N\frac{r_m}{(k-k_0)^m},\hspace{0.2cm}\text{in } k_0;\hspace{0.5cm}r(k)=\tilde{r}_0(k)+\sum_{m=1}^N\frac{(-1)^{m+1}r_m}{(k+k_0)^m}\hspace{0.2cm}\text{in } -k_0;
\end{equation}
\begin{equation}
r^*(k^*)=r_0^*(k^*)+\sum_{m=1}^N\frac{r_m^*}{(k-k_0^*)^m}\hspace{0.2cm}\text{in }k_0^*;\hspace{0.5cm}r^*(k^*)=\tilde{r}_0^*(k^*)+\sum_{m=1}^N\frac{(-1)^{m+1}r_m^*}{(k+k_0^*)}\hspace{0.2cm}\text{in }-k_0^*;
\end{equation}
\end{subequations}
where $r_m$ are defined by
\begin{align}
r_m=\lim_{k\rightarrow k_0}\frac{1}{(N-m)!}\frac{\partial^{N-m}}{\partial k^{N-m}}[(k-k_0)^Nr(k)],\hspace{0.5cm}m=1,2,...,N,
\end{align}
and $r_0(k)$ and $\tilde{r}_0(k)$ are analytic for all $k\in D^+$. The definition of $M(x,t,k)$ yields that $k=\pm k_0$ are Nth-order poles of $M_{12}$, while $k=\pm k_0^*$ are Nth-order poles of $M_{11}$. According to the normalization condition sated in proposition 5 one can set
\begin{subequations}
\begin{align}
&M_{11}(x,t,k)=1+\sum_{s=1}^N\Big(\frac{F_s(x,t)}{(k-k_0^*)^s}+\frac{H_s(x,t)}{(k+k_0^*)^s}\Big),\label{M01}\\
&M_{12}(x,t,k)=\sum_{s=1}^N\Big(\frac{G_s(x,t)}{(k-k_0)^s}+\frac{L_s(x,t)}{(k+k_0)^s}\Big),\label{M02}
\end{align}\label{M0}
\end{subequations}
where $F_s(x,t)$, $H_s(x,t)$, $G_s(x,t)$, $L_s(x,t)$($s=1,2,...,N$)  are unknown functions which need to be determined. Once these functions are solved, the solution $M(x,t,k)$ of RHP will be obtained and the solutions $q(x,t)$ of the GI equation will be obtained from (\ref{q0}).

Now we are in position to solve $F_s(x,t)$, $H_s(x,t)$, $G_s(x,t)$ and $L_s(x,t)$($s=1,2,...,N$). According to Taylor series expansion, one has
\begin{subequations}
\begin{equation}
e^{-2i\theta(k)}=\sum_{l=0}^{+\infty}f_l(x,t)(k-k_0)^l,\hspace{0.5cm}e^{-2i\theta(k)}=\sum_{l=0}^{+\infty}(-1)^lf_l(x,t)(k+k_0)^l,
\end{equation}
\begin{equation}
e^{2i\theta(k)}=\sum_{l=0}^{+\infty}f_l^*(x,t)(k-k_0^*)^l,\hspace{0.5cm}e^{2i\theta(k)}=\sum_{l=0}^{+\infty}(-1)^lf_l^*(x,t)(k+k_0^*)^l,
\end{equation}
\begin{equation}
M_{11}(x,t,k)=\sum_{l=0}^{+\infty}\mu_l(x,t)(k-k_0)^l,\hspace{0.5cm}M_{11}(x,t,k)=\sum_{l=0}^{+\infty}(-1)^l\mu_l(x,t)(k+k_0)^l,
\end{equation}
\begin{equation}
M_{12}(x,t,k)=\sum_{l=0}^{+\infty}\zeta_l(x,t)(k-k_0^*)^l,\hspace{0.5cm}M_{12}(x,t,k)=\sum_{l=0}^{+\infty}(-1)^{l+1}\zeta_l(x,t)(k+k_0^*)^l,
\end{equation}
\end{subequations}
where
\begin{subequations}
\begin{align}
&f_l(x,t)=\lim_{k\rightarrow k_0}\frac{1}{l!}\frac{\partial^l}{\partial k^l}e^{-2ik^2(x+2k^2t)};  \\
&\mu_l(x,t)=\lim_{k\rightarrow k_0}\frac{1}{l!}\frac{\partial^l}{\partial k^l}M_{11}(x,t,k),\hspace{0.5cm}\zeta_l(x,t)=\lim_{k\rightarrow k_0^*}\frac{1}{l!}\frac{\partial^l}{\partial k^l}M_{12}(x,t,k).\label{mu0zeta0}
\end{align}
\end{subequations}
When $k\in D^+$, we have the expansions in $k=k_0$
\begin{align}
M_{11}(k)=\phi_{-,11}=\sum_{l=0}^{+\infty}\mu_l(x,t)(k-k_0)^l,\hspace{0.5cm}&M_{12}(k)=\frac{\phi_{+,22}}{s_{22}}=e^{-2i\theta}r(k)\phi_{-,11}+\phi_{-,12}.
\end{align}
comparing the coefficients of $(k-k_0)^{-s}$ with (\ref{M02}), we can get
\begin{equation}
G_s(x,t)=\sum_{j=s}^N\sum_{l=0}^{j-s}r_jf_{j-s-l}(x,t)\mu_l(x,t).
\end{equation}
Similarly, from the expansions in $k=-k_0$, we can get that
\begin{equation}
L_s(x,t)=\sum_{j=s}^N\sum_{l=0}^{j-s}(-1)^{s+1}r_jf_{j-s-l}(x,t)\mu_l(x,t).
\end{equation}
With the same method, when $k\in D^-$, we can obtain that
\begin{equation}
F_s(x,t)=-\sum_{j=s}^N\sum_{l=0}^{j-s}r_j^*f_{j-s-l}^*(x,t)\zeta_l(x,t),\hspace{0.5cm}H_s(x,t)=\sum_{j=s}^{N}\sum_{l=0}^{j-s}(-1)^{s+1}r_j^*f_{j-s-l}^*(x,t)\zeta_l(x,t).\label{FH01}
\end{equation}
Actually, $\mu_l(x,t)$ and $\zeta_l(x,t)$ can also be expressed by $F_s(x,t)$, $H_s(x,t)$, $G_s(x,t)$ and $L_s(x,t)$. Recalling the definitions of $\zeta_l(x,t)$ and $\mu_l(x,t)$ given by (\ref{mu0zeta0}) and substituting (\ref{M0}) into them, we can obtain
\begin{align}
&\zeta_l(x,t)=\sum_{s=1}^N\left(\begin{array}{c}
                                 s+l-1\\
                                 l
                                 \end{array}\right)\Big\{\frac{(-1)^lG_s(x,t)}{(k_0^*-k_0)^{l+s}}+\frac{(-1)^lL_s(x,t)}{(k_0^*+k_0)^{l+s}}\Big\},\hspace{0.5cm}l=0,1,2,...,\\
&\mu_l(x,t)=\left\{\begin{array}{l}
                   1+\sum_{s=1}^N\Big\{\frac{F_s(x,t)}{(k_0-k_0^*)^s}+\frac{H_s(x,t)}{(k_0+k_0^*)^s}\Big\},\hspace{0.5cm}l=0;\\
                   \sum_{s=1}^N\left(\begin{array}{c}
                   s+l-1\\
                   l\end{array}\right)\Big\{\frac{(-1)^lF_s(x,t)}{(k_0-k_0^*)^{s+l}}+\frac{(-1)^lH_s(x,t)}{(k_0+k_0^*)^{s+l}}\Big\},\hspace{0.5cm}l=1,2,3,...
                   \end{array}\right.
\end{align}\label{zeta1mu1}
Using (\ref{FH01}) and (\ref{zeta1mu1}), we obtain the system
\begin{subequations}
\begin{align}
&F_s(x,t)=-\sum_{j=s}^N\sum_{l=0}^{j-s}\sum_{p=1}^{N}\left(\begin{array}{c}
                                                           p+l-1\\
                                                           l\end{array}\right)r_j^*f_{j-s-l}^*\Big\{\frac{(-1)^lG_p(x,t)}{(k_0^*-k_0)^{l+p}}+\frac{(-1)^lL_p(x,t)}{(k_0^*+k_0)^{l+p}}\Big\},\\
&H_s(x,t)=\sum_{j=s}^N\sum_{l=0}^{j-s}\sum_{p=1}^{N}(-1)^{s+1}\left(\begin{array}{c}
                                                           p+l-1\\
                                                           l\end{array}\right)r_j^*f_{j-s-l}^*\Big\{\frac{(-1)^lG_p(x,t)}{(k_0^*-k_0)^{l+p}}+\frac{(-1)^lL_p(x,t)}{(k_0^*+k_0)^{l+p}}\Big\},\\
&G_s(x,t)=\sum_{j=s}^Nr_jf_{j-s}+\sum_{j=s}^N\sum_{l=0}^{j-s}\sum_{p=1}^N\left(\begin{array}{c}
                                                                                p+l-1\\
                                                                                l\end{array}\right)r_jf_{j-s-l}\Big\{\frac{(-1)^lF_p(x,t)}{(k_0-k_0^*)^{l+p}}+\frac{(-1)^lH_p(x,t)}{(k_0+k_0^*)^{l+p}}\Big\},\\
&L_s(x,t)=\sum_{j=s}^N(-1)^{s+1}r_jf_{j-s}&\\&+\sum_{j=s}^N\sum_{l=0}^{j-s}\sum_{p=1}^N(-1)^{s+1}\left(\begin{array}{c}
                                                                                p+l-1\\
                                                                                l\end{array}\right)r_jf_{j-s-l}\Big\{\frac{(-1)^lF_p(x,t)}{(k_0-k_0^*)^{l+p}}+\frac{(-1)^lH_p(x,t)}{(k_0+k_0^*)^{l+p}}\Big\},\nonumber
\end{align}\label{system01}
\end{subequations}
Let us introduce
\begin{align}
&|\eta\rangle=(\eta_1,...,\eta_N)^T,\hspace{0.2cm}\eta_s=\sum_{j=s}^Nr_jf_{j-s}(x,t),\\
&|\tilde{\eta}\rangle=(\tilde{\eta}_1,...,\tilde{\eta}_N)^T,\hspace{0.2cm}\tilde{\eta}_s=\sum_{j=s}^N(-1)^{s+1}r_jf_{j-s}(x,t),\\
&|F\rangle=(F_1,F_2,...,F_N)^T,\hspace{0.5cm}|H\rangle=(H_1,H_2,...,H_N)^T,\\
&|G\rangle=(G_1,G_2,...,G_N)^T,\hspace{0.5cm}|L\rangle=(L_1,L_2,...,L_N)^T,\\
&\Omega_1=[\Omega_{1,sp}]_{N\times N}=\Big[-\sum_{j=s}^N\sum_{l=0}^{j-s}\left(\begin{array}{c}
                                                                               p+l-1\\
                                                                               l\end{array}\right)\frac{(-1)^lr_j^*f_{j-s-l}^*(x,t)}{(k_0^*-k_0)^{l+p}}\Big]_{N\times N},\\
&\Omega_2=[\Omega_{2,sp}]_{N\times N}=\Big[-\sum_{j=s}^N\sum_{l=0}^{j-s}\left(\begin{array}{c}
                                                                               p+l-1\\
                                                                               l\end{array}\right)\frac{(-1)^lr_j^*f_{j-s-l}^*(x,t)}{(k_0^*+k_0)^{l+p}}\Big]_{N\times N},\\
&\Omega_3=[\Omega_{3,sp}]_{N\times N}=\Big[\sum_{j=s}^N\sum_{l=0}^{j-s}\left(\begin{array}{c}
                                                                               p+l-1\\
                                                                               l\end{array}\right)\frac{(-1)^{s+l+1}r_j^*f_{j-s-l}^*(x,t)}{(k_0^*-k_0)^{l+p}}\Big]_{N\times N},\\
&\Omega_4=[\Omega_{4,sp}]_{N\times N}=\Big[\sum_{j=s}^N\sum_{l=0}^{j-s}\left(\begin{array}{c}
                                                                               p+l-1\\
                                                                               l\end{array}\right)\frac{(-1)^{s+l+1}r_j^*f_{j-s-l}^*(x,t)}{(k_0^*+k_0)^{l+p}}\Big]_{N\times N},
\end{align}
where the superscript $^T$ denotes the transposed matrix. Thus the linear system (\ref{system01}) can be rewritten as
\begin{equation}
\left\{\begin{array}{l}
        I|F\rangle+\mathbf{0}|H\rangle-\Omega_1|G\rangle-\Omega_2|L\rangle=\mathbf{0}\\
        \mathbf{0}|F\rangle+I|H\rangle-\Omega_3|G\rangle-\Omega_4|L\rangle=\mathbf{0}\\
        \Omega_1^*|F\rangle+\Omega_2^*|H\rangle+I|G\rangle+\mathbf{0}|L\rangle=|\eta\rangle\\
        \Omega_3^*|F\rangle+\Omega_4^*|H\rangle-\mathbf{0}|G\rangle-I|L\rangle=-|\tilde{\eta}\rangle
        \end{array}\right.
\end{equation}
Through direct calculations, $(|F\rangle,|H\rangle)^T$ and $(|G\rangle,|L\rangle)^T$ are explicitly solved as
\begin{align}
&(|F\rangle,|H\rangle)^T=\Omega(I_\sigma+\Omega^*\Omega)^{-1})(|\eta\rangle,-|\tilde{\eta})^T,\\
&(|G\rangle,|L\rangle)^T=(I_\sigma+\Omega^*\Omega)^{-1})(|\eta\rangle,-|\tilde{\eta})^T,
\end{align}
where
\begin{equation}
\Omega=\left(\begin{array}{cc}
             \Omega_1&\Omega_2\\
             \Omega_3&\Omega_4
             \end{array}\right),\hspace{0.5cm}I_\sigma=\left(\begin{array}{cc}
                                                              I_{N\times N}&\\
                                                              &-I_{N\times N}
                                                              \end{array}\right).
\end{equation}
Substituting $(|F\rangle,|H\rangle)^T$ and $(|G\rangle,|L\rangle)^T$ into the expansions of $M_{11}(x,t,k)$ and $M_{12}(x,t,k)$ given by (\ref{M0}), since it is well known that for matrix $A_{m\times n}$, $B_{n\times m}$ $\in \mathbb{K}$ ($\mathbb{K}$ is a field of numbers) $\det(I_m+AB)=\det(I_n+BA)$ ($I_m$ and $I_n$ are $m$ and $n$ dimension identity matrix respectively),  we get that
\begin{align}
\begin{split}
M_{11}(x,t,k)=&1+\big(\langle Y|,\langle\tilde{Y}|)(|F\rangle,|H\rangle\big)^T\\
&=\det\big(1+(\langle Y|,\langle\tilde{Y}|)(|F\rangle,|H\rangle)^T\big)\\
&=\det\big(1+(\langle Y|,\langle\tilde{Y}|)\Omega(I_\sigma+\Omega^*\Omega)^{-1}(|\eta\rangle,-|\tilde{\eta}\rangle)^T\big)\\
&=\det\big(I+(|\eta\rangle,-|\tilde{\eta}\rangle)^T(\langle Y|,\langle\tilde{Y}|)\Omega(I_\sigma+\Omega^*\Omega)^{-1}\big)\\
&=\frac{\det\big(I_\sigma+\Omega^*\Omega+(|\eta\rangle,-|\tilde{\eta}\rangle)^T(\langle Y|,\langle\tilde{Y}|)\Omega\big)}{\det\big(I_\sigma+\Omega^*\Omega\big)},
\end{split}
\end{align}
where
\begin{equation}
\langle Y(k)|=(\frac{1}{k-k_0^*},\frac{1}{(k-k_0^*)^2},...,\frac{1}{(k-k_0^*)^N}),\hspace{0.2cm}\langle \tilde{Y}(k)|=(\frac{1}{k+k_0^*},\frac{1}{(k+k_0^*)^2},...,\frac{1}{(k+k_0^*)^N}).
\end{equation}
In the same way, we get that
\begin{equation}
M_{12}=\frac{\det\Big(I_\sigma+\Omega^*\Omega+(|\eta\rangle,-|\tilde{\eta}\rangle)^T(\langle Y^*(k^*)|,\langle\tilde{Y}^*(k^*)|)\Big)}{\det\big(I_\sigma+\Omega^*\Omega\big)}-1.
\end{equation}
\begin{theorem}
With the rapidly decaying initial condition (\ref{bc0}), the $N$th order soliton of GI equation is
\begin{equation}
q(x,t)=2i\Big[\frac{\det\Big(I_\sigma+\Omega^*\Omega+(|\eta\rangle,-|\tilde{\eta}\rangle)^T(\langle Y_0|,\langle Y_0|)\Big)}{\det\big(I_\sigma+\Omega^*\Omega\big)}-1\Big],\label{q01}
\end{equation}
where,
\begin{equation}
\langle Y_0|=(1,0,...,0)_{1\times N}.
\end{equation}
\end{theorem}
\begin{proof}
From the expansion of $M_{12}(x,t,k)$, it follows that
\begin{align*}
q(x,t)&=\lim_{k\rightarrow\infty}2ikM_{12}(x,t,k)\\
&=\lim_{k\rightarrow\infty}2ik\big(\langle Y^*(k^*)|,\langle\tilde{Y}^*(k^*)|\big)\big(|G\rangle,|L\rangle\big)^T\\
&=\lim_{k\rightarrow\infty}2ik\big(\langle Y^*(k^*)|,\langle\tilde{Y}^*(k^*)|\big)\big(I_\sigma+\Omega^*\Omega\big)^{-1}\big(|\eta\rangle,-|\tilde{\eta}\rangle\big)^T\\
&=2i\big(\langle Y_0|,\langle Y_0|\big)\big(I_\sigma+\Omega^*\Omega\big)^{-1}\big(|\eta\rangle,-|\tilde{\eta}\rangle\big)^T\\
&=2i\Big[\frac{\det\Big(I_\sigma+\Omega^*\Omega+(|\eta\rangle,-|\tilde{\eta}\rangle)^T(\langle Y_0|,\langle Y_0|)\Big)}{\det\big(I_\sigma+\Omega^*\Omega\big)}-1\Big].
\end{align*}
\end{proof}

\subsection{  Multiple  high-order pole solutions}
Now we will study the general case that $s_{22}(k)$ has $N$ high-order zero points $k_1$, $k_2$,...,$k_N$, $k_i\in D^+$ for $i=1,2,...,N$, and their powers are $n_1$, $n_2$,...,$n_N$ respectively.  Let $r_{j}(k)$ be $r(k)'s$ Laurent series in $k=\pm k_j$, like the case of one high-order pole discussed above, we can obtain
\begin{align}
&r_{j}(k)=r_{j,0}(k)+\sum_{m_j=1}^{n_j}\frac{r_{j,m_j}}{(k-k_j)^{m_j}},\hspace{0.5cm}r_{j}^*(k^*)=r_{j,0}^*(k^*)+\sum_{m_j=1}^{n_j}\frac{r_{j,m_j}^*}{(k-k_j^*)^{m_j}},\\
&r_{j}(k)=\tilde{r}_{j,0}(k)+\sum_{m_j=1}^{n_j}\frac{(-1)^{m_j+1}r_{j,m_j}}{(k+k_j)^{m_j}},\hspace{0.5cm}r_{j}^*(k^*)=\tilde{r}_{j,0}^*(k^*)+\sum_{m_j=1}^{n_j}\frac{(-1)^{m_j+1}r_{j,m_j}^*}{(k+k_j^*)^{m_j}},
\end{align}
where
\begin{equation*}
r_{j,m_j}=\lim_{k\rightarrow k_j}\frac{1}{(n_j-m_j)!}\frac{\partial^{n_j-m_j}}{\partial k^{n_j-m_j}}\big[(k-k_j)^{n_j}r(k)\big],
\end{equation*}
and $r_{j,0}(k)$ ($j=1,...,N$) is analytic for all $k\in D^+$.

By the similar method in above, the multiple solitons of the GI equation are obtained as follows.
\begin{theorem}
With the rapidly decaying initial condition (\ref{bc0}), if $s_{22}(k)$ has $N$ distinct high-order poles, then the multiple solitons of GI equation have the same form as (\ref{q01})
\begin{equation}
q(x,t)=2i\Big[\frac{\det\Big(I_\sigma+\Omega^*\Omega+|\eta\rangle\langle Y_0|\Big)}{\det\big(I_\sigma+\Omega^*\Omega\big)}-1\Big],
\end{equation}
where
\begin{subequations}
\begin{align}
&|\eta\rangle=[|\eta_1\rangle,|\eta_2\rangle,...,|\eta_N\rangle]^T,\hspace{0.2cm}|\eta_j\rangle=[\eta_{j,1},\eta_{j,2},...,\eta_{j,n_j},-\tilde{\eta}_{j,1},-\tilde{\eta}_{j,2},...,-\tilde{\eta}_{j,n_j}],\\
&\eta_{j,l}=\sum_{m_j=l}^{n_j}r_{j,m_j}f_{m_j-l}(x,t),\hspace{0.5cm}\tilde{\eta}_{j,l}=\sum_{m_j=l}^{n_j}(-1)^{l+1}r_{j,m_j}f_{m_j-l}(x,t),\\
&\langle Y_0|=[\langle Y_1^0|,\langle Y_2^0|,...,\langle Y_N^0|],\hspace{0.5cm}\langle Y_j^0|=[\langle Y_j^{00}|,\langle Y_j^{00}|],\hspace{0.5cm}\langle Y_j^{00}|=[1,0,...,0]_{1\times n_j},\\
&\Omega=\left(\begin{array}{cccc}
              [\omega_{11}]&[\omega_{12}]&\cdots&[\omega_{1N}]\\{}
              [\omega_{21}]&[\omega_{22}]&\cdots&[\omega_{2N}]\\
              \vdots &\vdots &\ddots &\vdots\\{}
              [\omega_{N1}]&[\omega_{N2}]&\cdots&[\omega_{NN}]
              \end{array}\right),\hspace{0.5cm}[\omega_{jl}]_{2n_j\times 2n_l}=\left(\begin{array}{cc}
                                                                                     [w_{jl}^1]_{n_j\times n_l}&[\omega_{jl}^2]_{n_j\times n_l}\\{}
                                                                                     [w_{jl}^3]_{n_j\times n_l}&[\omega_{jl}^4]_{n_j\times n_l}
                                                                                     \end{array}\right),\\
&w_{jl,pq}^1=-\sum_{m_j=p}^{n_j}\sum_{s_j=0}^{m_j-p}\left(\begin{array}{c}
                                                            q+s_j-1\\
                                                            s_j\end{array}\right)\frac{(-1)^{s_j}r_{j,m_j}^*f_{j,m_j-p-s_j}^*}{(k_j^*-k_l)^{s_j+q}},\\
&w_{jl,pq}^2=-\sum_{m_j=p}^{n_j}\sum_{s_j=0}^{m_j-p}\left(\begin{array}{c}
                                                            q+s_j-1\\
                                                            s_j\end{array}\right)\frac{(-1)^{s_j}r_{j,m_j}^*f_{j,m_j-p-s_j}^*}{(k_j^*+k_l)^{s_j+q}},\\
&w_{jl,pq}^3=\sum_{m_j=p}^{n_j}\sum_{s_j=0}^{m_j-p}\left(\begin{array}{c}
                                                            q+s_j-1\\
                                                            s_j\end{array}\right)\frac{(-1)^{p+s_j+1}r_{j,m_j}^*f_{j,m_j-p-s_j}^*}{(k_j^*-k_l)^{s_j+q}},\\
&w_{jl,pq}^4=\sum_{m_j=p}^{n_j}\sum_{s_j=0}^{m_j-p}\left(\begin{array}{c}
                                                            q+s_j-1\\
                                                            s_j\end{array}\right)\frac{(-1)^{p+s_j+1}r_{j,m_j}^*f_{j,m_j-p-s_j}^*}{(k_j^*+k_l)^{s_j+q}},\\
&I_\sigma=\left(\begin{array}{ccccc}
                 I_{n_1\times n_1}& & & & \\
                 &-I_{n_1\times n_1}& & & \\
                  & &\ddots & &\\
                  & & &I_{n_N\times n_N}& \\
                  & & & &-I_{n_N\times n_N}
                  \end{array}\right).
\end{align}
\end{subequations}
\end{theorem}

We then give the figure of one-soliton solution with one second-order pole
\begin{figure}[H]
    \centering
    \subfigure{
    \includegraphics[width=0.46\textwidth]{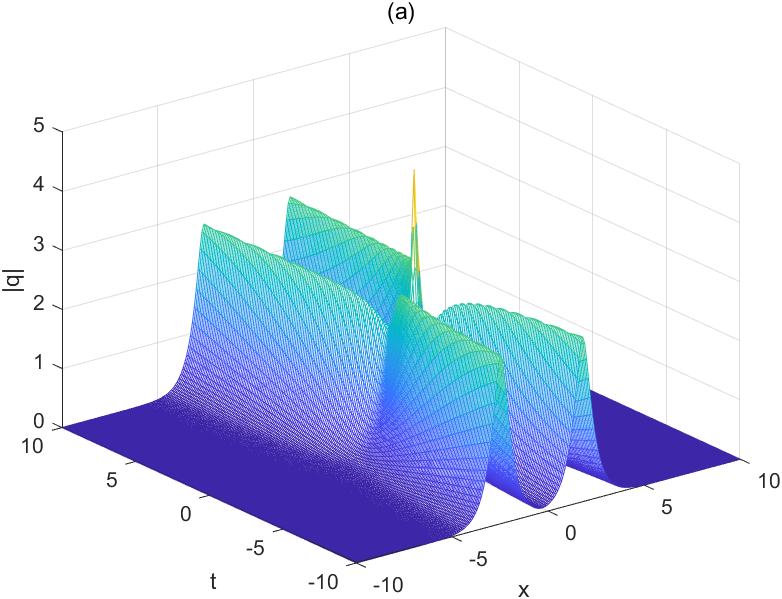}
   }
    \subfigure{
    \includegraphics[width=0.4\textwidth]{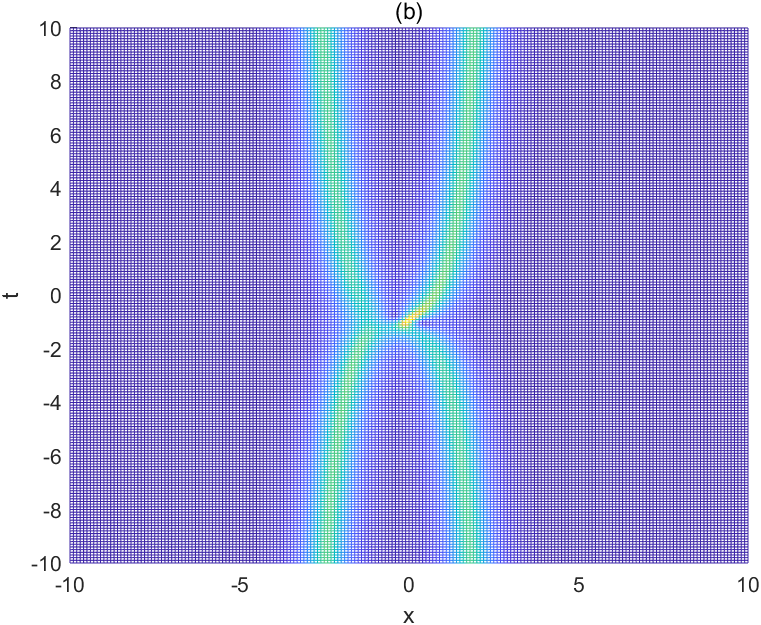}
    }
\flushleft{\footnotesize {\bf Figure 2.} One-soliton with one second-order pole, here  taking  parameters $r_1=1$, $r_2=2$, $k_0=e^{\frac{\pi i}{4}}$, $k_0^*=e^{-\frac{\pi i}{4}}$. (a): The three-dimensional graph. (b): The contour of the wave.

 }
    \end{figure}

\section{IST with nonzero boundary and high-order poles}

\subsection{ Riemann surface and uniformization variable}

To make convenience for the later calculation, we handle the Lax pair (\ref{lax1}) and the boundary condition (\ref{bc1}) at the beginning. We make a proper transformation
\begin{align*}
&q \rightarrow qe^{-\frac{3}{2}iq_0^4t+iq_0^2x},\\
&\phi \rightarrow e^{(-\frac{3}{4}iq_0^4t+\frac{1}{2}iq_0^2x)\sigma_3}\phi.
\end{align*}
The GI equation (\ref{GI}) then becomes
\begin{equation}
iq_t+q_{xx}+2iq_0^2q_x-iq^2q_x^\ast-q_0^2q^2q^*+\frac{1}{2}q^3q^{*2}+\frac{1}{2}q_0^4q=0,\label{GI1}
\end{equation}
with corresponding boundary
\begin{equation}
\lim_{x \to \pm\infty}q(x,t)=q_\pm, \label{boundary}
\end{equation}
where $|q_\pm|=q_0$.

The GI equation (\ref{GI1}) is the compatibility condition of the Lax pair
\begin{equation}
\phi_x=X\phi,\hspace{0.5cm}\phi_t=T\phi,\label{lax2}
\end{equation}
where
\begin{align*}
X=-ik^2\sigma_3+\frac{i}{2}(|q|^2-q_0^2)\sigma_3+kQ,\hspace{0.5cm}Q=\left(\begin{array}{cc}
                                                                           0 & q \\
                                                                           -q^* & 0
                                                                          \end{array}\right),
\end{align*}
\begin{align*}
T=&-2ik^4\sigma_3+(ik^2|q|^2-iq_0^2|q|^2+\frac{i}{4}|q|^4+\frac{3}{4}iq_0^4)\sigma_3+\frac{1}{2}(Q_xQ-QQ_x)\\
&+2k^3Q-ikQ_x\sigma_3-kq_0^2Q.
\end{align*}

Under the  boundary (\ref{boundary}),   asymptotic spectral problem of the Lax pair (\ref{lax2}) becomes
\begin{equation}
\phi_x=X_\pm\phi,\hspace{0.5cm}\phi_t=T_\pm\phi,\label{asymptoticlax}
\end{equation}
where
\begin{equation}
X_\pm=-ik^2\sigma_3+kQ_\pm,\hspace{0.5cm}T_\pm=(2k^2-q_0^2)X_\pm,
\end{equation}
and
\begin{equation*}
Q_\pm=\left(\begin{array}{cc}
             0 & q_\pm \\
             -q_\pm^* & 0
            \end{array}\right).
\end{equation*}
The eigenvalues of the matrix $X_\pm$ are $\pm ik\lambda$, where $ \lambda^2=k^2+q_0^2 $. Since the eigenvalues are doubly branched, we introduce the two-sheeted Riemann surface defined by
\begin{equation}
\lambda^2=k^2+q_0^2,
\end{equation}
then  $\lambda(k)$ is single-valued on this surface. The branch points are  $k=\pm iq_0$. Letting
$$ k+iq_0=r_1e^{i\theta_1}, \ \   k-iq_0=r_2e^{i\theta_2},$$
we can get two single-valued analytic functions on the Riemann surface
\begin{equation}
\lambda(k)=\Bigg\{\begin{array}{ll}
                 \text{$(r_1r_2)^{1/2}e^{i(\theta_1+\theta_2)/2},$} &\text{on $S_1,$}\\\\
                 \text{$-(r_1r_2)^{1/2}e^{i(\theta_1+\theta_2)/2},$} &\text{on $S_2,$} \end{array}
\label{tranlambda}
\end{equation}
where  $-\pi/2 < \theta_j <3/2\pi $ for $j=1,2$.

 Gluing the two copies of the complex plane $S_1$ and $S_2$ along the segment $[-iq_0,iq_0]$, we then obtain the Riemann surface. Along the real $k$ axis we have $\lambda(k)=\pm {\rm sign}(k)\sqrt{k^2+q_0^2}$, where the "$\pm$" applies on $S_1$ and $S_2$ of the Riemann surface respectively, and where the square root sign denotes the principal branch of the real-valued square root function.

Next, we take a uniformization variable
\begin{equation}
z=k+\lambda,
\end{equation}
then we obtain two single-valued functions
\begin{equation}
k(z)=\frac{1}{2}(z-\frac{q_0^2}{z}),\hspace{0.5cm}\lambda(z)=\frac{1}{2}(z+\frac{q_0^2}{z}).\label{uniformization55}
\end{equation}
This implies that we  can  discuss the scattering problem on a standard $z$-plane instead of the two-sheeted Riemann surface by the inverse mapping.
We define $D^+$, $D^-$ and $\Sigma$ on $z$-plane as
\begin{equation*}
\Sigma=\mathbb{R} \cup i\mathbb{R} \backslash \{0\},\hspace{0.5cm}D^+=\{z:{\rm Re}z{\rm Im}z>0\},\hspace{0.5cm}D^-=\{z:{\rm Re}z{\rm Im}z<0\}.
\end{equation*}
the two domains are shown in Figure 3.

 From these discussions, we can derive that
\begin{align*}
{\rm Im}(k(z)\lambda(z))&= {\rm Im}\frac{z^4-q_0^4}{4z^2}= {\rm Im}\frac{(|z|^4+q_0^4)z^2-2q_0^4(({\rm Re}z)^2-({\rm Im}z)^2)}{4|z|^4}\\
                  &= \frac{1}{4|z|^4}(|z|^4+q_0^4){\rm Im}z^2= \frac{1}{2|z|^4}(|z|^4+q_0^4){\rm Re}z{\rm Im}z,
\end{align*}
which implies  that
\begin{equation}
{\rm Im}(k(z)\lambda(z))\Bigg\{\begin{array}{ll}
                 \text{$=0,$} &\text{as $z\in\Sigma$}\\
                 \text{$>0,$} &\text{as $z\in D^+$} \hspace{0.5cm}\text{.}\\
                 \text{$<0,$} &\text{as $z\in D^-$}
                 \end{array}
\end{equation}

\begin{figure}[H]
	\centering
		\begin{tikzpicture}[node distance=2cm]
		\filldraw[violet!30,line width=2] (2.4,0.01) rectangle (0.01,2.4);
		\filldraw[violet!30,line width=2] (-2.4,-0.01) rectangle (-0.01,-2.4);
		\draw[->](-3,0)--(3,0)node[right]{$\mathbb{R}$};
		\draw[->](0,-3)--(0,3)node[above]{$i\mathbb{R}$};
		\draw[->](0,0)--(-0.8,0);
		\draw[->](-0.8,0)--(-1.8,0);
		\draw[->](0,0)--(0.8,0);
		\draw[->](0.8,0)--(1.8,0);
		\coordinate (A) at (0.5,1.2);
		\coordinate (B) at (0.6,-1.2);
		\coordinate (G) at (-0.6,1.2);
		\coordinate (H) at (-0.5,-1.2);
		\coordinate (I) at (0.16,0);
		\fill (A) circle (0pt) node[right] {$D^+$};
		\fill (B) circle (0pt) node[right] {$D_-$};
		\fill (G) circle (0pt) node[left] {$D_-$};
		\fill (H) circle (0pt) node[left] {$D^+$};
		\fill (I) circle (0pt) node[below] {$0$};
	\end{tikzpicture}
	\flushleft{\footnotesize {\bf Figure 3.} Complex $z$-plane consist of  the region $D^+$ (the violet regions) and  the $D^-$ (the white regions).}
\end{figure}

\subsection{Spectral Analysis}

\subsubsection{ Eigenfunctions and scattering matrix}

\quad  For eigenvalue $\pm i\lambda$, we can write the asymptotic eigenvector matrix as
\begin{equation}
Y_\pm=\left(\begin{array}{cc}
            1 & -\frac{iq_\pm}{z} \\
            -\frac{iq_\pm^*}{z} & 1
            \end{array}\right) =I-\frac{i}{z}\sigma_3Q_\pm,
\end{equation}
so that $X_\pm$ and $T_\pm$ can be diagonalized by $Y_\pm$
\begin{equation}
X_\pm=Y_\pm(-ik\lambda\sigma_3)Y_\pm^{-1},\hspace{0.5cm}T_\pm=Y_\pm(-(2k^2-q_0^2)ik\lambda\sigma_3)Y_\pm^{-1}.\label{lax3}
\end{equation}
Direct computation shows  that
\begin{equation}
\det (Y_\pm)=1+\frac{q_0^2}{z^2}\triangleq\gamma,
\end{equation}
and
\begin{equation}
Y_\pm^{-1}=\frac{1}{\gamma}\left(\begin{array}{cc}
                                  1 & \frac{iq_\pm}{z} \\
                                  \frac{iq_\pm^*}{z} & 1
                                 \end{array}\right)=\frac{1}{\gamma}(I+\frac{i}{z}\sigma_3Q_\pm),\hspace{0.5cm}z\neq\pm iq_0.
\end{equation}
 Substituting (\ref{lax3}) into (\ref{asymptoticlax}), we immediately obtain
\begin{equation}
(Y_\pm^{-1}\psi)_x=-ik\lambda\sigma_3(Y_\pm^{-1}\psi),\hspace{0.5cm}(Y_\pm^{-1}\psi)_t=-(2k^2-q_0^2)ik\lambda\sigma_3(Y_\pm^{-1}\psi),\hspace{0.5cm}z\neq\pm iq_0,
\end{equation}
from which we can derive the solution of the asymptotic spectral problem (\ref{asymptoticlax})
\begin{equation}
\psi(x,t,z)=\Bigg\{\begin{array}{ll}
                    Y_\pm e^{i\theta(z)\sigma_3},& z\neq\pm iq_0,\\
                    I+(x-3q_0^2t)Y_\pm(z),& z=\pm iq_0,
                    \end{array}
\end{equation}
where
\begin{equation*}
\theta(x,t,z)=-k(z)\lambda(z)[x+(2k^2(z)-q_0^2)t].
\end{equation*}
For conveniens, we will omit $x$ and $t$ dependence in $\theta(x,t,z)$ henceforth.

We define the Jost eigenfunctions $\phi_\pm(x,t,z)$ as the simultaneous solutions of both parts of the Lax pair so that
\begin{equation}
\phi_\pm=Y_\pm e^{i\theta(z)\sigma_3}+o(1),\hspace{0.5cm}x \rightarrow \pm\infty.\label{asymptoticphi}
\end{equation}
We introduce modified eigenfunctions by factorizing the asymptotic exponential oscillations
\begin{equation}
\mu_\pm=\phi_\pm e^{-i\theta(z)\sigma_3},\label{defmu}
\end{equation}
then we have
\begin{equation*}
\mu_\pm \sim Y_\pm, \hspace{0.5cm} x \rightarrow \pm\infty.
\end{equation*}
Meanwhile,  $\mu_\pm$    acquire the equivalent Lax pair
\begin{align}
&(Y_\pm^{-1}\mu_\pm)_x-ik\lambda[Y_\pm^{-1}\mu_\pm,\sigma_3]=Y_\pm^{-1}\Delta X_\pm\mu_\pm, \label{eqlax1}\\
&(Y_\pm^{-1}\mu_\pm)_t-ik\lambda(2k^2-q_0^2)[Y_\pm^{-1}\mu_\pm,\sigma_3]=Y_\pm^{-1}\Delta T_\pm\mu_\pm,\label{eqlax2}
\end{align}
where $\Delta X_\pm=X-X_\pm$ and $\Delta T_\pm=T-T_\pm$. These two equations can be written in full derivative form
\begin{equation}
d(e^{-i\theta(z){\hat{\sigma}}_3}Y_\pm^{-1}\mu_\pm) = e^{-i\theta(z){\hat{\sigma}}_3}[Y_\pm^{-1}(\Delta X_\pm dx+\Delta T_\pm dt)\mu_\pm],
\end{equation}
 which leads to the Volterra integral equations
\begin{equation}
\mu_\pm(x,t,z)=\left\{\begin{array}{ll}
                       Y_\pm+\int_{\pm\infty}^{x}Y_\pm e^{-ik\lambda(x-y){\hat{\sigma}}_3}[Y_\pm^{-1}\Delta X_\pm(y,t)\mu_\pm(y,t,z)]dy,& z\neq\pm iq_0,\\[6pt]
                      Y_\pm+\int_{\pm\infty}^{x}[I+(x-y)X_\pm(z)]\Delta X_\pm(y,t)\mu_\pm(y,t,z)dy,& z=\pm iq_0,
                       \end{array}\label{jost}
                \right.
\end{equation}
where  we define  $e^{\alpha\hat{\sigma}}A:=e^{\alpha\sigma}Ae^{-\alpha\sigma}$, for a matrix $A$.

\quad  Since ${\rm tr}X={\rm tr}T=0$  in  (\ref{lax2}),  then by using Abel formula, we have
\begin{equation*}
(\det\phi_\pm)_x=(\det\phi_\pm)_t=0,\hspace{0.5cm}\det(\mu_\pm)=\det(\phi_\pm e^{-i\theta(z)\sigma_3})=\det(\phi_\pm).
\end{equation*}
So that $(\det\mu_\pm)_x=(\det\mu_\pm)_t=0$, which means $\det(\mu_\pm)$ is independent with $x,t$.  Furthermore, we know that $\mu_\pm$ is invertible from
\begin{equation}
\det\mu_\pm=\lim_{x \to \pm\infty}\det(\mu_\pm)=\det Y_\pm=\gamma\neq0, \hspace{0.5cm}x,t\in\mathbb{R},\hspace{0.5cm}z\in\Sigma_0.\label{detphi}
\end{equation}

 Since   $\phi_\pm$ are two fundamental matrix solutions of the linear Lax  pair (\ref{lax2}),  there exists a relation between $\phi_+$ and $\phi_-$
\begin{equation}
\phi_+(x,t,z)=\phi_-(x,t,z)S(z),\hspace{0.5cm}x,t\in\mathbb{R},\hspace{0.5cm}z\in\Sigma_0,\label{scattering}
\end{equation}
where $S(z)$ is called scattering matrix and (\ref{detphi}) implies that $\det S(z)=1$. Letting $S(z)=(s_{ij})$, for the individual columns
\begin{equation}
\phi_{+,1}=s_{11}\phi_{-,1}+s_{21}\phi_{-,2},\hspace{0.5cm}\phi_{+,2}=s_{12}\phi_{-,1}+s_{22}\phi_{-,2}.\label{idvsca}
\end{equation}

By using (\ref{scattering}),  we obtain
\begin{align}
s_{11}(z)=\frac{{\rm Wr}(\phi_{+,1},\phi_{-,2})}{\gamma},\hspace{0.5cm}s_{12}(z)=\frac{{\rm Wr}(\phi_{+,2},\phi_{-,2})}{\gamma},\label{scatteringcoefficient1}\\
s_{21}(z)=\frac{{\rm Wr}(\phi_{-,1},\phi_{+,1})}{\gamma},\hspace{0.5cm}s_{22}(z)=\frac{{\rm Wr}(\phi_{-,1},\phi_{+,2})}{\gamma}.\label{scatteringcoefficient2}
\end{align}

\subsubsection{Analyticity}

Here we directly  state   the   analyticity of eigenfunctions $\mu_{\pm}$ and scattering data $s_{11}, \ s_{22}$,   the detail proofs of them   were given in our paper \cite{RN31}.

\begin{proposition}
Suppose $q(x,t)-q_\pm\in L^1(\mathbb{R}^\pm)$, then the Volterra integral equation (\ref{jost}) has unique solutions $\mu_\pm(x,t,z)$ defined by (\ref{defmu}) in $\Sigma_0:=\Sigma\setminus \{\pm iq_0\}$. Moreover, the columns $\mu_{-,1}$ and $\mu_{+,2}$ can be analytically extended to $D^+$ and continuously extended to $D^+\cup\Sigma_0$, while the columns $\mu_{+,1}$ and $\mu_{-,2}$ can be analytically extended to $D^-$ and continuously extended to $D^-\cup\Sigma_0$, where $\mu_{\pm,j}(x,t,z)(j=1,2)$ denote the $j$-th column of $\mu_\pm$.
\end{proposition}

\begin{proposition}
Suppose $(1+|x|)(q(x,t)-q_\pm)\in L^1(\mathbb{R}^\pm)$, then the Volterra integral equation (\ref{jost}) has unique solutions $\mu_\pm(x,t,z)$ defined by (\ref{defmu}) in $\Sigma$. Besides, the columns $\mu_{-,1}$ and $\mu_{+,2}$ can be analytically extended to $D^+$ and continuously extended to $D^+\cup\Sigma$, while the columns $\mu_{+,1}$ and $\mu_{-,2}$ can be analytically extended to $D^-$ and continuously extended to $D^-\cup\Sigma$.
\end{proposition}

\begin{lemma}
Consider an n-dimensional first-order homogeneous linear ordinary differential equation, $dy(x)/dx = A(x) y(x)$, on an interval $\mathbb{D}\in\mathbb{R}$, where $A(x)$ denotes a complex square matrix of order n. Let $\Phi$ be a matrix-valued solution of this equation. If the trace $\text{tr} A(x)$ is a continuous function, then one has
\begin{equation}
\det\Phi(x)=\det\Phi(x_0)\exp\Big[\int_{x_0}^x  \text{tr} A(\xi)d\xi\Big],\hspace{0.5cm}x,x_0\in\mathbb{\mathbb{D}}.
\end{equation}
\end{lemma}
\begin{proposition}
The Jost solutions $\Phi(x,t,z)$ are the simultaneous solutions of both parts of the Lax pair (\ref{lax2}).
\end{proposition}

\begin{proposition}
Suppose $q(x,t)-q_\pm\in L^1(\mathbb{R}^\pm)$. Then $s_{11}$ can be analytically extended to $D^-$ and continuously extended to $D^-\cup\Sigma_0$, while $s_{22}$ can be analytically extended to $D^+$ and continuously extended to $D^+\cup\Sigma_0$. Moreover, $s_{12}$ and $s_{21}$ are continuous in $\Sigma_0$.
\end{proposition}

Note that we can not exclude the possible existence of zeros for $s_{11}(z)$ and $s_{22}(z)$ along $\Sigma_0$. To solve the RH problem, we restrict our consideration to potentials without spectral singularities, i.e., $s_{11}(z)\neq0$, $s_{22}(z)\neq0$ for $z\in\Sigma$. Besides, we assume that the scattering coefficients are continuous at the branch points. The reflection coefficients which will be needed in the inverse problem are
\begin{equation}
\tilde{r}(z)=\frac{s_{21}}{s_{11}},\hspace{0.5cm}r(z)=\frac{s_{12}}{s_{22}}.
\end{equation}

\subsubsection{Symmetries}
\quad For the GI equation with nonzero boundary, we not only need to deal with the map $k\mapsto k^*$, but also need to pay attention to the sheets of the Riemann surface. We can see from the Riemann surface that the transformation $z\mapsto z^*$ implies $(k,\lambda)\mapsto(k^*,\lambda^*)$ and $z\mapsto-q_0^2/z$ implies $(x,\lambda)\mapsto(k,-\lambda)$. Therefore, we would like to discuss the symmetries in the following way.
\begin{proposition}
The Jost solution, scattering matrix and reflection coefficients satisfy the following reduction conditions on $z$-plane

$\bullet$ The first symmetry reduction
\begin{equation}
\phi_\pm(x,t,z)=\sigma_2\phi_\pm^*(x,t,z^*)\sigma_2,\hspace{0.5cm}S(z)=\sigma_2S^*(z^*)\sigma_2,\hspace{0.5cm}r(z)=-\tilde{r}^*(z^*),
\label{symmetry1}
\end{equation}
where $\sigma_2=\left(\begin{array}{cc}
                            0 & -i \\
                            i & 0
                            \end{array}\right).$

$\bullet$ The second symmetry reduction
\begin{equation}
\phi_\pm(x,t,z)=\sigma_1\phi_\pm^*(x,t,-z^*)\sigma_1,\hspace{0.5cm}S(z)=\sigma_1S^*(-z^*)\sigma_1,\hspace{0.5cm}r(z)=\tilde{r}^*(-z^*),\label{symmetry2}
\end{equation}
where $\sigma_1=\left(\begin{array}{cc}
                            0 & 1 \\
                            1 & 0
                            \end{array}\right).$

$\bullet$ The third symmetry reduction
\begin{align}
&\phi_\pm(x,t,z)=-\frac{i}{z}\phi_\pm(x,t,-\frac{q_0^2}{z})\sigma_3Q_\pm, \label{symmetry3.1}\\
&S(z)=(\sigma_3Q_-)^{-1}S(-\frac{q_0^2}{z})\sigma_3Q_+,\hspace{0.5cm}r^*(z^*)=\frac{q_-}{q_-^*}\tilde{r}(-\frac{q_0^2}{z}).\label{symmetry3.2}
\end{align}
\end{proposition}

\subsubsection{Asymptotic behaviors}

\quad To solve the RH problem in the next section, it is necessary to discuss the asymptotic behaviors of the modified Jost solutions and scattering matrix as $z \rightarrow \infty$ and $z \rightarrow 0$ by the standard Wentzel-Kramers-Brillouin (WKB) expansions.
\begin{proposition}
The asymptotic behaviors for the modified Jost solutions are given as
\begin{align}
&\mu_\pm(x,t,z)=I+o(z^{-1}),\hspace{0.5cm}z \rightarrow \infty,\label{asymptotic1}\\
&\mu_\pm(x,t,z)=-\frac{i}{z}\sigma_3Q_\pm+o(1),\hspace{0.5cm}z \rightarrow 0.
\end{align}
\end{proposition}
From  (\ref{asymptotic1}) we can get that
\begin{equation}
q(x,t)=\lim_{z\rightarrow\infty}iz\mu_\pm^{(12)},\label{q11}
\end{equation}
which will be used in the following.
Inserting the above asymptotic behaviors for the modified Jost eigenfunctions into the Wronskian representation (\ref{scatteringcoefficient1}) and (\ref{scatteringcoefficient2}), with a little calculations, we get the asymptotic behaviors of the scattering matrix.

\begin{proposition}
The asymptotic behaviors of the scattering matrix are
\begin{align}
&S(z)=I+O(z^{-1}),\hspace{0.5cm}z \rightarrow \infty,\label{asympsca1}\\
&S(z)={\rm diag}(\frac{q_-}{q_+},\frac{q_+}{q_-})+O(z),\hspace{0.5cm}z \rightarrow 0.\label{asympsca2}
\end{align}
\end{proposition}

\subsubsection{Distribution of  spectrum}

\quad The discrete spectrum of the scattering problem is the set of all values $z\in \mathbb{C}\setminus\Sigma$, for which eigenfunctions exist in $L^2(\mathbb{R})$. We would like to show that these values are the zeros of $s_{11}(z)$ in $D^-$ and those of $s_{22}(z)$ in $D^+$.

We can show that  the  uniformization transformation (\ref{uniformization55})  changes  the segment $[-iq_0,iq_0]$  on $k$-plane into
the circle $|z|=q_0$ on $z$-plane.
We suppose that $s_{22}$ has one $N$th order zero $z_0$ in $D^+\cap\{z\in\mathbb{C}:{\rm Im}z>0,|z|>q_0\}$,
then symmetries  (\ref{symmetry1})-(\ref{symmetry3.2}) imply that
\begin{equation}
s_{22}(\pm z_0)=0 \Leftrightarrow s_{11}^*(\pm z_0^*)=0 \Leftrightarrow s_{11}(\pm\frac{q_0^2}{z_0})=0 \Leftrightarrow s_{22}(\pm\frac{q_0^2}{z_0^*})=0.\label{quartetsym}
\end{equation}
Therefore, the discrete spectrum is the set
\begin{equation}
Z=\left\{\pm z_0, \pm z_0^*,\pm \frac{q_0^2}{z_0}, \pm\frac{q_0^2}{z_0^*}\right\},
\end{equation}
which can be seen in Figure 4.
\begin{center}
\begin{tikzpicture}[node distance=2cm]
\filldraw[violet!30,line width=3] (2.8,0.01) rectangle (0.01,2.8);
\filldraw[violet!30,line width=3] (-2.8,-0.01) rectangle (-0.01,-2.8);
\draw[->](-3,0)--(3,0)node[right]{Re$z$};
\draw[->](0,-3)--(0,3)node[above]{Im$z$};
\draw (2,0) arc (0:360:2);
\draw[->](0,0)--(-1.5,0);
\draw[->](-1.5,0)--(-2.8,0);
\draw[->](0,0)--(1.5,0);
\draw[->](1.5,0)--(2.8,0);
\draw[->](0,2.7)--(0,2.2);
\draw[->](0,1.6)--(0,0.8);
\draw[->](0,-2.7)--(0,-2.2);
\draw[->](0,-1.6)--(0,-0.8);
\coordinate (A) at (2.2,2.2);
\coordinate (B) at (2.2,-2.2);
\coordinate (C) at (-0.9090909090909,0.9090909090909);
\coordinate (D) at (-0.9090909090909,-0.9090909090909);
\coordinate (E) at (0.9090909090909,0.9090909090909);
\coordinate (F) at (0.9090909090909,-0.9090909090909);
\coordinate (G) at (-2.2,2.2);
\coordinate (H) at (-2.2,-2.2);
\coordinate (I) at (0,2);
\coordinate (J) at (1.414213562373095,1.414213562373095);
\coordinate (K) at (1.414213562373095,-1.414213562373095);
\coordinate (L) at (-1.414213562373095,1.414213562373095);
\coordinate (M) at (-1.414213562373095,-1.414213562373095);
\fill (A) circle (1pt) node[right] {$z_n$};
\fill (B) circle (1pt) node[right] {$z_n^*$};
\fill (C) circle (1pt) node[right] {$-\frac{q_0^2}{z_n}$};
\fill (D) circle (1pt) node[right] {$-\frac{q_0^2}{z_n^*}$};
\fill (E) circle (1pt) node[left] {$\frac{q_0^2}{z_n^*}$};
\fill (F) circle (1pt) node[left] {$\frac{q_0^2}{z_n}$};
\fill (G) circle (1pt) node[left] {$-z_n^*$};
\fill (H) circle (1pt) node[left] {$-z_n$};
\fill (I) circle (1pt) node[right] {$q_0$};
\label{zplane2}
\end{tikzpicture}
\flushleft{\footnotesize {\bf Figure 4.}  Distribution of   the the discrete spectrum and  the contours for the RH problem on complex $z$-plane.
 }
\end{center}

\subsection{Riemann-Hilbert Problem}

\quad As we all know, the equation (\ref{scattering}) is the beginning of the formulation of the inverse problem. We always regard it as a relation between eigenfunctions analytic in $D^+$ and those analytic in $D^-$. Thus, it is necessary for us to introduce the following RH problem.
\begin{proposition}
Define the sectionally meromorphic matrix
\begin{equation}
M(x,t,z)=\Bigg\{\begin{array}{ll}
            M^-=\left(\begin{array}{cc}
           \dfrac{\mu_{+,1}}{s_{11}} & \mu_{-,2}\\
           \end{array}\right), &\text{as } z\in D^-,\\
                 M^+=\left(\begin{array}{cc}
                           \mu_{-,1} & \dfrac{\mu_{+,2}}{s_{22}}\\
                           \end{array}\right), &\text{as }z\in D^+.\\
                 \end{array}
\end{equation}
Then a multiplicative matrix RH problem is proposed:

$\bullet$ Analyticity: $M(x,t,z)$ is analytic in $\mathbb{C}\setminus\Sigma$ and has single poles.

$\bullet$ Jump condition
\begin{equation}
M^-(x,t,z)=M^+(x,t,z)(I-G(x,t,z)),\hspace{0.5cm}z \in \Sigma,\label{jump}
\end{equation}
where
\begin{equation}
G(x,t,z)=\left(\begin{array}{cc}
                r(z)\tilde{r}(z) & e^{2i\theta}r(z)\\
                -e^{-2i\theta}\tilde{r}(z) & 0
                \end{array}\right).
\end{equation}

$\bullet$ Asymptotic behaviors
\begin{align}
&M(x,t,z) \sim I+O(z^{-1}),\hspace{0.5cm}z \rightarrow \infty,\label{asymbehv1}\\
&M(x,t,z) \sim -\frac{i}{z}\sigma_3Q_-+O(1),\hspace{0.5cm}z \rightarrow 0.\label{asymbehv2}
\end{align}
\end{proposition}
From (\ref{q11}) we know that
\begin{equation}
q(x,t)=\lim_{z\rightarrow\infty}izM^{(12)}.
\end{equation}

\subsection{Single  high-order pole solutions}
Let $z_0\in D^+$ is the Nth-order pole, from the symmetries (\ref{symmetry1})-(\ref{symmetry3.1}) it is obvious that $-z_0$, $\pm\frac{q_0^2}{z_0^*}$ $\in D^+$ also is the Nth-order pole of $s_{22}(z)$. Then $\pm z_0^*$ and $\pm\frac{q_0^2}{z_0}$ are the Nth-order poles of $s_{11}(z)$. The discrete spectrum is the set
\begin{equation}
\{\pm z_0,\pm z_0^*,\pm\frac{q_0^2}{z_0^*},\pm\frac{q_0^2}{z_0}\},
\end{equation}
which can be seen in Figure 4

Let $\nu_1=z_0$, $\nu_2=\frac{q_0^2}{z_0^*}$, then the discrete spectrum is $\{\pm\nu_1,\pm\nu_2,\pm\nu_1^*,\pm\nu_2^*\}$. Let
\begin{equation}
s_{22}(z)=(z^2-\nu_1^2)^N(z^2-\nu_2^2)^Ns_0(z),
\end{equation}
in which $s_0(z)\neq 0$ in $D^+$. According to the Laurent series expansion in poles, $r(z)$ and $r^*(z^*)$ can be respectively expanded as
\begin{subequations}
\begin{align}
&r_j(z)=r_{0,j}(z)+\sum_{m_j=1}^N\frac{r_{j,m_j}}{(z-\nu_j)^{m_j}},\hspace{0.5cm}\text{in } z=\nu_j,\hspace{0.5cm}j=1,2;\\
&r_j(z)=\tilde{r}_{0,j}(z)+\sum_{m_j=1}^N\frac{(-1)^{m_j+1}r_{j,m_j}}{(z+\nu_j)^{m_j}}\hspace{0.5cm}\text{in } z=-\nu_j,\hspace{0.5cm}j=1,2;\\
&r^*_j(z^*)=r_{0,j}^*(z^*)+\sum_{m_j=1}^N\frac{r_{j,m_j}^*}{(z-\nu_j^*)^{m_j}}\hspace{0.5cm}\text{in }z=\nu_j^*,\hspace{0.5cm}j=1,2;\\
&r^*_j(z^*)=\tilde{r}_{0,j}^*(z^*)+\sum_{m_j=1}^N\frac{(-1)^{m_j+1}r_{j,m_j}^*}{(z+\nu_j^*)}\hspace{0.5cm}\text{in }z=-\nu_j^*,\hspace{0.5cm}j=1,2,
\end{align}
\end{subequations}
where $r_{j,m_j}$ are defined by
\begin{align}
&r_{j,m_j}=\lim_{z\rightarrow \nu_j}\frac{1}{(N-m_j)!}\frac{\partial^{N-m_j}}{\partial z^{N-m_j}}[(z-\nu_j)^Nr_j(z)],\hspace{0.5cm}m_j=1,2,...,N.
\end{align}
and $r_{0,j}(z)$ and $\tilde{r}_{0,j}(z)$ are analytic for all $z\in D^+$. The definition of $M(x,t,k)$ yields that $z=\pm\nu_j$ ($j=1,2$) are Nth-order poles of $M_{12}$, while $z=\pm\nu_j^*$ ($j=1,2$) are Nth-order poles of $M_{11}$. According to the normalization condition sated in proposition 13 one can set
\begin{subequations}
\begin{align}
&M_{11}(x,t,z)=1+\sum_{j=1}^2\sum_{s=1}^N\Big(\frac{F_{j,s}(x,t)}{(z-\nu_j^*)^s}+\frac{H_{j,s}(x,t)}{(z+\nu_j^*)^s}\Big),\label{M11}\\
&M_{12}(x,t,z)=-\frac{i}{z}q_-+\sum_{j=1}^2\sum_{s=1}^N\Big(\frac{G_{j,s}(x,t)}{(z-\nu_j)^s}+\frac{L_{j,s}(x,t)}{(z+\nu_j)^s}\Big),\label{M12}
\end{align}\label{M1}
\end{subequations}
where $F_{j,s}(x,t)$, $H_{j,s}(x,t)$, $G_{j,s}(x,t)$, $L_{j,s}(x,t)$($s=1,2,...,N$, $j=1,2$)  are unknown functions which need to be determined. Once these functions are solved, the solution $M(x,t,z)$ of RHP will be obtained and the solutions $q(x,t)$ of the GI equation will be obtained from (\ref{M1}).

Now we are in position to solve $F_{j,s}(x,t)$, $H_{j,s}(x,t)$, $G_{j,s}(x,t)$ and $L_{j,s}(x,t)$($s=1,2,...,N$, $j=1,2$). According to Taylor series expansion, one has
\begin{subequations}
\begin{equation}
e^{2i\theta(z)}=\sum_{l=0}^{+\infty}f_{j,l}(x,t)(z-\nu_j)^l,\hspace{0.5cm}e^{2i\theta(z)}=\sum_{l=0}^{+\infty}(-1)^lf_{j,l}(x,t)(z+\nu_j)^l,
\end{equation}
\begin{equation}
e^{-2i\theta(z)}=\sum_{l=0}^{+\infty}f_{j,l}^*(x,t)(z-\nu_j^*)^l,\hspace{0.5cm}e^{-2i\theta(z)}=\sum_{l=0}^{+\infty}(-1)^lf_{j,l}^*(x,t)(z+\nu_j^*)^l,
\end{equation}
\begin{equation}
M_{11}(x,t,z)=\sum_{l=0}^{+\infty}\mu_{j,l}(x,t)(z-\nu_j)^l,\hspace{0.5cm}M_{11}(x,t,z)=\sum_{l=0}^{+\infty}(-1)^l\mu_{j,l}(x,t)(z+\nu_j)^l,
\end{equation}
\begin{equation}
M_{12}(x,t,z)=\sum_{l=0}^{+\infty}\zeta_{j,l}(x,t)(z-\nu_j^*)^l,\hspace{0.5cm}M_{12}(x,t,z)=\sum_{l=0}^{+\infty}(-1)^{l+1}\zeta_{j,l}(x,t)(z+\nu_j^*)^l,
\end{equation}
\end{subequations}
where
\begin{subequations}
\begin{align}
&f_{j,l}(x,t)=\lim_{z\rightarrow \nu_j}\frac{1}{l!}\frac{\partial^l}{\partial z^l}e^{2i\theta(z)},\\
&\mu_{j,l}(x,t)=\lim_{z\rightarrow \nu_j}\frac{1}{l!}\frac{\partial^l}{\partial z^l}M_{11}(x,t,z),\hspace{0.5cm}\zeta_{j,l}(x,t)=\lim_{z\rightarrow \nu_j^*}\frac{1}{l!}\frac{\partial^l}{\partial z^l}M_{12}(x,t,z).\label{mu1zeta1}
\end{align}
\end{subequations}
When $z\in D^+$, we have the expansions in $z=\nu_j$ ($j=1,2$)
\begin{align}
&M_{11}(z)=\mu_{-,11}=\sum_{l=0}^{+\infty}\mu_{j,l}(x,t)(z-\nu_j)^l,\\
&M_{12}(z)=\frac{\mu_{+,22}(x,t,z)}{s_{22}(z)}=e^{2i\theta}r(z)\mu_{-,11}(x,t,z)+\mu_{-,12}(x,t,z)
\end{align}
comparing the coefficients of $(z-\nu_j)^{-s}$ with (\ref{M12}), we can get
\begin{equation}
G_{j,s}(x,t)=\sum_{m_j=s}^N\sum_{l=0}^{m_j-s}r_{j,m_j}f_{j,m_j-s-l}(x,t)\mu_{j,l}(x,t).\label{G1}
\end{equation}
Similarly, from the expansions in $z=-\nu_j$ ($j=1,2$), we can get that
\begin{equation}
L_{j,s}(x,t)=\sum_{m_j=s}^N\sum_{l=0}^{m_j-s}(-1)^{s+1}r_{j,m_j}f_{j,m_j-s-l}(x,t)\mu_{j,l}(x,t).\label{L1}
\end{equation}
By the same method, when $z\in D^-$, we can obtain that
\begin{align}
&F_{j,s}(x,t)=-\sum_{m_j=s}^N\sum_{l=0}^{m_j-s}r_{j,m_j}^*f_{j,m_j-s-l}^*(x,t)\zeta_{j,l}(x,t),\label{F1}\\
&H_{j,s}(x,t)=\sum_{m_j=s}^{N}\sum_{l=0}^{m_j-s}(-1)^{s+1}r_{j,m_j}^*f_{j,m_j-s-l}^*(x,t)\zeta_{j,l}(x,t).\label{H1}
\end{align}
Actually, $\mu_{j,l}(x,t)$ and $\zeta_{j,l}(x,t)$ ($j=1,2$) can also be expressed by $F_{j,s}(x,t)$, $H_{j,s}(x,t)$, $G_{j,s}(x,t)$ and $L_{j,s}(x,t)$ ($j=1,2$). Recalling the definitions of $\zeta_{j,l}(x,t)$ and $\mu_{j,l}(x,t)$ ($j=1,2$) given by (\ref{mu1zeta1}) and substituting (\ref{M1}) into them, we can obtain
\begin{subequations}
\begin{align}
&\zeta_{j,l}(x,t)=\frac{(-1)^{l+1}}{(\nu_j^*)^{l+1}}iq_-+\sum_{p=1}^2\sum_{s=1}^N\left(\begin{array}{c}
                                 s+l-1\\
                                 l
                                 \end{array}\right)\Big\{\frac{(-1)^lG_{p,s}(x,t)}{(\nu_j^*-\nu_p)^{l+s}}+\frac{(-1)^lL_{p,s}(x,t)}{(\nu_j^*+\nu_p)^{l+s}}\Big\},l=0,1,...,\\
&\mu_{j,l}(x,t)=\left\{\begin{array}{l}
                   1+\sum_{p=1}^2\sum_{s=1}^N\Big\{\frac{F_{p,s}(x,t)}{(\nu_j-\nu_p^*)^s}+\frac{H_{p,s}(x,t)}{(\nu_j+\nu_p^*)^s}\Big\},\hspace{0.5cm}l=0;\\
                   \sum_{p=1}^2\sum_{s=1}^N\left(\begin{array}{c}
                   s+l-1\\
                   l\end{array}\right)\Big\{\frac{(-1)^lF_{p,s}(x,t)}{(\nu_j-\nu_p^*)^{s+l}}+\frac{(-1)^lH_{p,s}(x,t)}{(\nu_j+\nu_p^*)^{s+l}}\Big\},\hspace{0.5cm}l=1,2,3,...
                   \end{array}\right.
\end{align}\label{zeta2mu2}
\end{subequations}
Using (\ref{G1})-(\ref{zeta2mu2}), we obtain the system
\begin{subequations}
\begin{align}
\begin{split}
&F_{j,s}(x,t)=-iq_-\sum_{m_j=s}^N\sum_{l=0}^{m_j-s}\frac{(-1)^{l+1}}{(\nu_j^*)^{l+1}}r_{j,m_j}^*f_{j,m_j-s-l}^*(x,t)\\&-\sum_{m_j=s}^N\sum_{l=0}^{m_j-s}\sum_{p=1}^2\sum_{q=1}^N\left(\begin{array}{c}
                                                           q+l-1\\
                                                           l\end{array}\right)r_{j,m_j}^*f_{j,m_j-s-l}^*\Big\{\frac{(-1)^lG_{p,q}(x,t)}{(\nu_j^*-\nu_p)^{l+q}}+\frac{(-1)^lL_{p,q}(x,t)}{(\nu_j^*+\nu_p)^{l+q}}\Big\},
\end{split}
\end{align}
\begin{align}
\begin{split}
&H_{j,s}(x,t)=iq_-\sum_{m_j=s}^N\sum_{l=0}^{m_j-s}\frac{(-1)^{s+l}}{(\nu_j^*)^{l+1}}r_{j,m_j}^*f_{j,m_j-s-l}^*(x,t)\\&+\sum_{m_j=s}^N\sum_{l=0}^{m_j-s}\sum_{p=1}^2\sum_{q=1}^N(-1)^{s+1}\left(\begin{array}{c}
                                                           q+l-1\\
                                                           l\end{array}\right)r_{j,m_j}^*f_{j,m_j-s-l}^*\Big\{\frac{(-1)^lG_{p,q}(x,t)}{(\nu_j^*-\nu_p)^{l+q}}+\frac{(-1)^lL_{p,q}(x,t)}{(\nu_j^*+\nu_p)^{l+q}}\Big\},
\end{split}
\end{align}
\begin{align}
\begin{split}
&G_{j,s}(x,t)=\sum_{m_j=s}^Nr_{j,m_j}f_{j,m_j-s}(x,t)\\&+\sum_{m_j=s}^N\sum_{l=0}^{m_j-s}\sum_{p=1}^2\sum_{q=1}^N\left(\begin{array}{c}
                                                                                q+l-1\\
                                                                                l\end{array}\right)r_{j,m_j}f_{j,m_j-s-l}\Big\{\frac{(-1)^lF_{p,q}(x,t)}{(\nu_j-\nu_p^*)^{l+q}}+\frac{(-1)^lH_{p,q}(x,t)}{(\nu_j+\nu_p^*)^{l+q}}\Big\},
\end{split}
\end{align}
\begin{align}
\begin{split}
&L_{j,s}(x,t)=\sum_{m_j=s}^N(-1)^{s+1}r_{j,m_j}f_{j,m_j-s}(x,t)\\&+\sum_{m_j=s}^N\sum_{l=0}^{m_j-s}\sum_{p=1}^2\sum_{q=1}^N(-1)^{s+1}\left(\begin{array}{c}
                                                                                q+l-1\\
                                                                                l\end{array}\right)r_{j,m_j}f_{j,m_j-s-l}\Big\{\frac{(-1)^lF_{p,q}(x,t)}{(\nu_j-\nu_p^*)^{l+q}}+\frac{(-1)^lH_{p,q}(x,t)}{(\nu_j+\nu_p^*)^{l+q}}\Big\},
\end{split}
\end{align}
\end{subequations}
Let us define
\begin{align*}
&|\eta_j\rangle=(\eta_{j1},...,\eta_{jN})^T,\hspace{0.5cm}\eta_{js}=-iq_-\sum_{m_j=s}^N\sum_{l=0}^{m_j-s}\frac{(-1)^{l+1}}{(\nu_j^*)^{l+1}}r_{j,m_j}^*f_{j,m_j-s-l}^*(x,t), \hspace{0.5cm}j=1,2;\\
&|\tilde{\eta}_j\rangle=(\tilde{\eta}_{j1},...,\tilde{\eta}_{jN})^T,\hspace{0.5cm}\tilde{\eta}_{js}=iq_-\sum_{m_j=s}^N\sum_{l=0}^{m_j-s}\frac{(-1)^{s+l}}{(\nu_j^*)^{l+1}}r_{j,m_j}^*f_{j,m_j-s-l}^*(x,t), \hspace{0.5cm}j=1,2;\\
&|\xi_j\rangle=(\xi_{j1},...,\xi_{jN})^T,\hspace{0.5cm}\xi_{js}=\sum_{m_j=s}^Nr_{j,m_j}f_{j,m_j-s}(x,t), \hspace{0.5cm}j=1,2;\\
&|\tilde{\xi}_j\rangle=(\tilde{\xi}_{j1},...,\tilde{\xi}_{jN})^T,\hspace{0.5cm}\tilde{\xi}_{js}=\sum_{m_j=s}^N(-1)^{s}r_{j,m_j}f_{j,m_j-s}(x,t), \hspace{0.5cm}j=1,2;
\end{align*}
\begin{align*}
&\Omega_{jp}=[\Omega_{jp}]_{sq}=-\sum_{m_j=s}^N\sum_{l=0}^{m_j-s}\left(\begin{array}{c}
                                                           q+l-1\\
                                                           l\end{array}\right)\frac{(-1)^lr_{j,m_j}^*f_{j,m_j-s-l}^*}{(\nu_j^*-\nu_p)^{l+q}},\hspace{0.5cm}j,p=1,2;\\
&\Omega_{j,p+2}=[\Omega_{j,p+2}]_{sq}=-\sum_{m_j=s}^N\sum_{l=0}^{m_j-s}\left(\begin{array}{c}
                                                           q+l-1\\
                                                           l\end{array}\right)\frac{(-1)^lr_{j,m_j}^*f_{j,m_j-s-l}^*}{(\nu_j^*+\nu_p)^{l+q}},\hspace{0.5cm}j,p=1,2,;\\
&\Omega_{j+2,p}=[\Omega_{j+2,p}]_{sq}=\sum_{m_j=s}^N\sum_{l=0}^{m_j-s}(-1)^{s+1}\left(\begin{array}{c}
                                                           q+l-1\\
                                                           l\end{array}\right)\frac{(-1)^lr_{j,m_j}^*f_{j,m_j-s-l}^*}{(\nu_j^*-\nu_p)^{l+q}},\hspace{0.5cm}j,p=1,2;\\
&\Omega_{j+2,p+2}=[\Omega_{j+2,p+2}]_{sq}=\sum_{m_j=s}^N\sum_{l=0}^{m_j-s}(-1)^{s+1}\left(\begin{array}{c}
                                                           q+l-1\\
                                                           l\end{array}\right)\frac{(-1)^lr_{j,m_j}^*f_{j,m_j-s-l}^*}{(\nu_j^*+\nu_p)^{l+q}},\hspace{0.5cm}j,p=1,2.
\end{align*}
\begin{align*}
&|F_p\rangle=(F_{p,1},...,F_{p,N})^T,\hspace{0.5cm}|H_p\rangle=(H_{p,1},...,H_{p,N})^T,\hspace{0.5cm}p=1,2;\\
&|G_p\rangle=(G_{p,1},...,G_{p,N})^T,\hspace{0.5cm}|L_p\rangle=(L_{p,1},...,L_{p,N})^T,\hspace{0.5cm}p=1,2.
\end{align*}
Let
\begin{equation}
\Omega=\left(\begin{array}{ccc}
             \Omega_{11}&\cdots&\Omega_{14}\\
             \vdots&\ddots&\vdots\\
             \Omega_{41}&\cdots&\Omega_{44}
             \end{array}\right);
\end{equation}
and
\begin{align*}
&|\alpha_1\rangle=(|\eta_1\rangle,|\eta_2\rangle,|\tilde{\eta}_1\rangle,|\tilde{\eta}_2\rangle)^T,\hspace{0.5cm}|\alpha_2\rangle=(|\xi_1\rangle,|\xi_2\rangle,|\tilde{\xi}_1\rangle,|\tilde{\xi}_2\rangle)^T;\\
&|K_1\rangle=(|F_1\rangle,|F_2\rangle,|H_1\rangle,|H_2\rangle)^T,\hspace{0.5cm}|K_2\rangle=(|G_1\rangle,|G_2\rangle,|L_1\rangle,|L_2\rangle)^T.
\end{align*}
Using the similar method with zero boundary condition, we have
\begin{equation}
|K_2\rangle=-\Omega^*(I_\sigma+\Omega^*\Omega)^{-1}\alpha_1+(I_\sigma+\Omega^*\Omega)^{-1}\alpha_2,
\end{equation}
where
\begin{equation*}
I_\sigma=\left(\begin{array}{cccc}
               I\\
               &I\\
               & &-I\\
               & & &-I
               \end{array}\right)_{4N\times4N},
\end{equation*}
so that
\begin{align}
\begin{split}
M_{12}(x,t,z)&=-\frac{i}{z}q_-+\langle Y|K_2\rangle\\
&=-\frac{i}{z}q_-+\langle Y|(-\Omega^*(I_\sigma+\Omega^*\Omega)^{-1}\alpha_1+(I_\sigma+\Omega^*\Omega)^{-1}\alpha_2)\\
&=-\frac{i}{z}q_-+\frac{\det(I_\sigma+\Omega^*\Omega+|\alpha_2\rangle\langle Y|)-\det(I_\sigma+\Omega^*\Omega+|\alpha_1\rangle\langle Y|\Omega^*)}{\det(I_\sigma+\Omega^*\Omega)},
\end{split}
\end{align}
where
\begin{equation*}
\langle Y|=\Big(\frac{1}{z-\nu_1},...,\frac{1}{(z-\nu_1)^N},\frac{1}{z-\nu_2},...,\frac{1}{(z-\nu_2)^N},\frac{1}{z+\nu_1},...,\frac{1}{(z+\nu_1)^N},\frac{1}{z+\nu_2},...,\frac{1}{(z+\nu_2)^N}\Big)_{1\times 4N}.
\end{equation*}
\begin{theorem}
With the nonzero boundary condition (\ref{bc1}), the $N$th order soliton of GI equation is
\begin{equation}
q(x,t)=q_-+i\Big[\frac{\det(I_\sigma+\Omega^*\Omega+|\alpha_2\rangle\langle Y_0|)-\det(I_\sigma+\Omega^*\Omega+|\alpha_1\rangle\langle Y_0|\Omega^*)}{\det(I_\sigma+\Omega^*\Omega)}\Big],\label{q1}
\end{equation}
where,
\begin{equation}
\langle Y_0|=(1,0,...,0,1,0,...,0,1,0,...,0,1,0,...,0)_{1\times 4N}.
\end{equation}
\end{theorem}

\subsection{ Multiple  high-order pole solutions}

Now we will study the general case that $s_{22}(z)$ has $N$ high-order zero points $z_1$, $z_2$,...,$z_N$, $z_k\in D^+$ for $k=1,2,...,N$, and their powers are $n_1$, $n_2$,...,$n_N$ respectively. Let $\nu_1^k=z_k$, $\nu_2^k=\frac{q_0^2}{z_k^*}$.  Let $r_{j}^k(z)$ be $r(z)'s$ Laurent series in $z=\nu_j^k$ ($j=1,2$), like the case of one high-order pole discussed above, we can obtain
\begin{align}
&r_{j}^k(z)=r_{j,0}^k(z)+\sum_{m_j=1}^{n_k}\frac{r_{j,m_j}^k}{(z-\nu_j^k)^{m_j}},\hspace{0.5cm}r_{j}^{*k}(z^*)=r_{j,0}^*(z^*)+\sum_{m_j=1}^{n_k}\frac{r_{j,m_j}^{k*}}{(z-\nu_j^{k*})^{m_j}},\\
&r_{j}^k(z)=\tilde{r}_{j,0}(z)+\sum_{m_j=1}^{n_k}\frac{(-1)^{m_j+1}r_{j,m_j}^k}{(z+\nu_j^k)^{m_j}},\hspace{0.5cm}r_{j}^{*k}(z^*)=\tilde{r}_{j,0}^*(z^*)+\sum_{m_j=1}^{n_k}\frac{(-1)^{m_j+1}r_{j,m_j}^{k*}}{(z+\nu_j^{k*})^{m_j}},
\end{align}
where
\begin{equation*}
r_{j,m_j}^k=\lim_{z\rightarrow \nu_j^k}\frac{1}{(n_k-m_j)!}\frac{\partial^{n_k-m_j}}{\partial k^{n_k-m_j}}\big[(z-\nu_j^k)^{n_k}r(z)\big],
\end{equation*}
and $r_{j,0}^k(z)$ ($k=1,...,N$) is analytic for all $z\in D^+$.

By the similar method in above, the multiple solitons of the GI equation are obtained as follows.
\begin{theorem}
With the nonzero boundary condition (\ref{bc1}), if $s_{22}(z)$ has $N$ distinct high-order poles, then the multiple solitons of GI equation have the same form as (\ref{q1})
\begin{equation}
q(x,t)=q_-+i\Big[\frac{\det(I_\sigma+\Omega^*\Omega+|\alpha_2\rangle\langle Y_0|)-\det(I_\sigma+\Omega^*\Omega+|\alpha_1\rangle\langle Y_0|\Omega^*)}{\det(I_\sigma+\Omega^*\Omega)}\Big],
\end{equation}
where
\begin{subequations}
\begin{align}
&|\alpha_1\rangle=\big(|\alpha_1^1\rangle,...,|\alpha_1^N\rangle\big)^T,\hspace{0.5cm}|\alpha_1^k\rangle=\big(|\eta_1^k\rangle,|\eta_2^k\rangle,|\tilde{\eta}_1^k\rangle,|\tilde{\eta}_2^k\rangle\big)^T,\hspace{0.1cm}k=1,...,N,\\
&|\alpha_2\rangle=\big(|\alpha_2^1\rangle,...,|\alpha_2^N\rangle\big)^T,\hspace{0.5cm}|\alpha_2^k\rangle=\big(|\xi_2^k\rangle,|\xi_2^k\rangle,|\tilde{\xi}_1^k\rangle,|\tilde{\xi}_2^k\rangle\big)^T,\hspace{0.1cm}k=1,...,N,\\
&|\eta_j^k\rangle=[|\eta_{j1}^k\rangle,...,|\eta_{jN}^k\rangle]^T,\hspace{0.5cm}|\xi_j^k\rangle=[|\xi_{j1}^k\rangle,...,|\xi_{jN}^k\rangle]^T,\hspace{0.1cm}j=1,2,\\
&\eta_{js}^k=-iq_-\sum_{m_j=s}^{n_k}\sum_{l=0}^{m_j-s}\frac{(-1)^{l+1}}{(\nu_j^{k*})^{l+1}}r_{j,m_j}^{k*}f_{j,m_j-s-l}^{k*}(x,t),\hspace{0.5cm}j=1,2;\\
&\tilde{\eta}_{js}^k=iq_-\sum_{m_j=s}^{n_k}\sum_{l=0}^{m_j-s}\frac{(-1)^{s+l}}{(\nu_j^{k*})^{l+1}}r_{j,m_j}^{k*}f_{j,m_j-s-l}^{k*}(x,t),\hspace{0.5cm}j=1,2;\\
&\xi_{js}^k=\sum_{m_j=s}^{n_k}r_{j,m_j}^kf_{j,m_j-s}^k(x,t),\hspace{0.5cm}\tilde{\xi}_{js}^k=-\sum_{m_j=s}^{n_k}(-1)^{s+1}r_{j,m_j}^kf_{j,m_j-s}^k(x,t),\hspace{0.2cm}j=1,2;\\
&\langle Y_0|=[\langle Y_1^0|,\langle Y_2^0|,...,\langle Y_N^0|],\hspace{0.5cm}\langle Y_k^0|=[1,0,...,0,1,0,...,0,1,0,...,0,1,0,...,0]_{1\times 4N},\\
&\Omega=\left(\begin{array}{cccc}
              [\omega_{11}]&[\omega_{12}]&\cdots&[\omega_{1N}]\\{}
              [\omega_{21}]&[\omega_{22}]&\cdots&[\omega_{2N}]\\
              \vdots &\vdots &\ddots &\vdots\\{}
              [\omega_{N1}]&[\omega_{N2}]&\cdots&[\omega_{NN}]
              \end{array}\right),\hspace{0.5cm}[\omega_{kh}]_{4n_k\times 4n_h}=\left(\begin{array}{ccc}
                                                                                      \omega_{kh}^{11}&\cdots&\omega_{kh}^{14}\\
                                                                                       \vdots&\ddots&\vdots\\
                                                                                      \omega_{kh}^{41}&\cdots&\omega_{kh}^{44}
                                                                                     \end{array}\right),\\
&\omega_{kh}^{jp}=[\omega_{kh}^{jp}]_{sq}=-\sum_{m_j=s}^{n_k}\sum_{l=0}^{m_j-s}\left(\begin{array}{c}
                                                           q+l-1\\
                                                           l\end{array}\right)\frac{(-1)^lr_{j,m_j}^{k*}f_{j,m_j-s-l}^{k*}}{(\nu_j^{k*}-\nu_p^k)^{l+q}},\hspace{0.5cm}j,p=1,2;\\
&\omega_{kh}^{j,p+2}=[\omega_{kh}^{j,p+2}]_{sq}=-\sum_{m_j=s}^{n_k}\sum_{l=0}^{m_j-s}\left(\begin{array}{c}
                                                           q+l-1\\
                                                           l\end{array}\right)\frac{(-1)^lr_{j,m_j}^{k*}f_{j,m_j-s-l}^{k*}}{(\nu_j^{k*}+\nu_p^k)^{l+q}},\hspace{0.5cm}j,p=1,2,;
\end{align}
\begin{align}
&\omega_{kh}^{j+2,p}=[\omega_{j+2,p}]_{sq}=\sum_{m_j=s}^{n_k}\sum_{l=0}^{m_j-s}(-1)^{s+1}\left(\begin{array}{c}
                                                           q+l-1\\
                                                           l\end{array}\right)\frac{(-1)^lr_{j,m_j}^{k*}f_{j,m_j-s-l}^{k*}}{(\nu_j^{k*}-\nu_p^k)^{l+q}},\hspace{0.5cm}j,p=1,2;\\
&\omega_{kh}^{j+2,p+2}=[\omega_{j+2,p+2}]_{sq}=\sum_{m_j=s}^{n_k}\sum_{l=0}^{m_j-s}(-1)^{s+1}\left(\begin{array}{c}
                                                           q+l-1\\
                                                           l\end{array}\right)\frac{(-1)^lr_{j,m_j}^{k*}f_{j,m_j-s-l}^{k*}}{(\nu_j^{k*}+\nu_p^k)^{l+q}},\hspace{0.5cm}j,p=1,2;\\
&I_\sigma=\left(\begin{array}{ccc}
                 I_{\sigma_1}\\
                 &\vdots&\\
                 &&I_{\sigma_N}
                  \end{array}\right),\hspace{0.5cm}I_{\sigma_k}=\left(\begin{array}{cccc}
                                                                      I\\
                                                                      &I\\
                                                                      &&-I\\
                                                                      &&&-I
                                                                      \end{array}\right)_{4n_k\times 4n_k},\hspace{0.1cm}k=1,...,N.
\end{align}
\end{subequations}
\end{theorem}

We then give the figures of one-soliton solution with one second-order pole
\begin{figure}[H]
    \centering
    \subfigure{
    \includegraphics[width=0.46\textwidth]{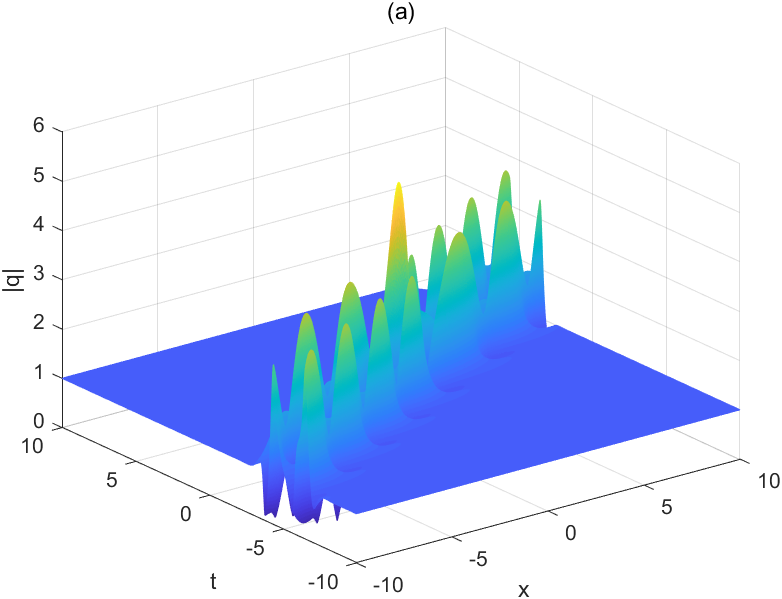}
   }
    \subfigure{
    \includegraphics[width=0.4\textwidth]{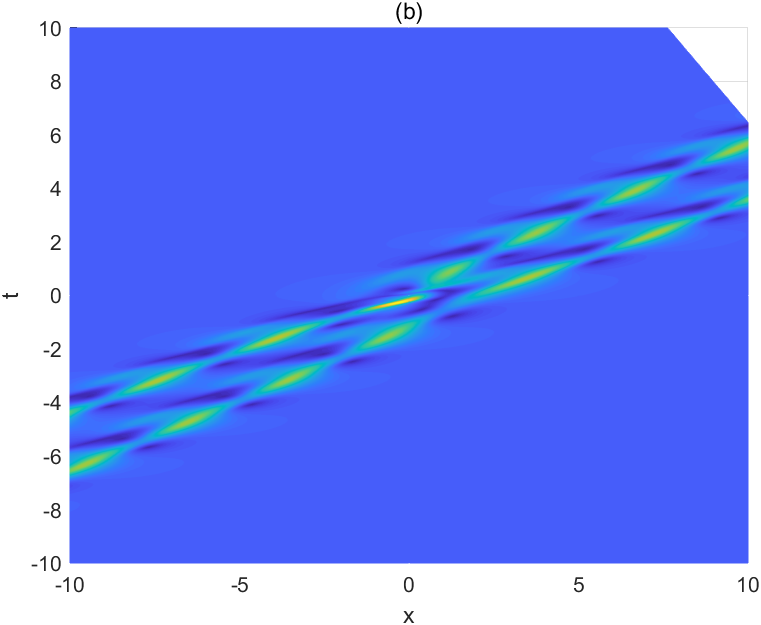}
    }
\flushleft{\footnotesize {\bf Figure 6.} One-soliton with one second-order pole, here taking  parameters $r_{11}=1$, $r_{12}=2$,$r_{21}=4$, $r_{22}=3$, $q_0=1$, $z_0=2e^{\frac{\pi i}{4}}$. (a): The three-dimensional graph. (b): The contour of the wave.

 }
    \end{figure}

\noindent\textbf{Acknowledgements}

This work is supported by  the National Science
Foundation of China (Grant No. 11671095,  51879045).

\hspace*{\parindent}
\\


\begin{thebibliography}{10}



\bibitem{RN19}
Gardner  C.S., Greene  J.M., Kruskal  M.D., Miura  R.M.,
\newblock {Method for solving the Korteweg-deVries equation}.
\newblock {\em Physical Review Letters}, 19(1967), 1095-1097.


\bibitem{RN16}
Ablowitz  M.J., Clarkson  P.A.,
\newblock {\em Soliton, Nonlinear Evolution Equations and Inverse Scattering}.
\newblock Cambridge University Press, Cambridge, 1991.

\bibitem{RN20}
Yang  J.K.,
\newblock {\em Nonlinear Waves in Integrable and Nonintegrable Systems}.
\newblock Philadelphia: Soc Indus Appl Math, 2010.

\bibitem{RN15}
Deift  P., Zhou  X.,
\newblock {A steepest descent method for oscillatory Riemann-Hilbert problems},
\newblock  {\em Annals of Mathematics}, 137(1993),  295-368.

\bibitem{RN17}
Biondini  G., Kova$\breve{c}$i$\breve{c}$  G.,
\newblock {Inverse scattering transform for the focusing nonlinear Schr\"{o}dinger equation with nonzero boundary conditions},
\newblock {\em Journal of Mathematical Physics}, 55(2014), 1-22.



\bibitem{RN18}
Pichler  M., Biondini  G.,
\newblock On the focusing non-linear Schr\"{o}dinger equation with non-zero boundary conditions and double poles.
\newblock {\em IMA Journal of Applied Mathematics}, 82(2017), 131-151.



\bibitem{RN21}
Biondini   G., Kraus  D.,
\newblock Inverse scattering transform for the defocusing Manakov system with nonzero boundary conditions.
\newblock {\em SIAM Journal on Mathematical Analysis}, 47(2015), 706-757.

\bibitem{RN22}
Prinari  B., Ablowitz  M.J.,  Biondini  G.,
\newblock Inverse scattering transform for the vector nonlinear Schr\"{o}dinger equation with nonvanishing boundary conditions.
\newblock {\em Journal of Mathematical Physics}, 47(2006), 1-32.

\bibitem{RN23}
Biondini  G., Kraus  D.K., Prinari  B.,
\newblock {The three-component defocusing nonlinear Schr\"{o}dinger equation with nonzero boundary conditions}.
\newblock {\em Communications in Mathematical Physics}, 348(2016), 475-533.





\bibitem{RN24}
Bilman D.,  Buckingham R.,
\newblock Large-Order Asymptotics for Multiple-Pole Solitons of the Focusing Nonlinear Schr$\ddot{o}$dinger Equation.
\newblock {\em Journal of Nonlinear Science}, 29(2019), 2185-2229.

\bibitem{RN25}
Zakharov V.E.,  Shabat A.B.,
\newblock Exact Theory of Two-dimensional Self-focusing and One-dimensional Self-modulating Waves in Nonlinear Media.
\newblock {\em Soviet Physics JETP}, 34(1972), 62-69.

\bibitem{RN26}
Kodama Y.,  Hasegawa A.,
\newblock Nonlinear Pulse Propagation in a Monomode Dielectric Guide.
\newblock {\em IEEE Journal of Quantum Electrics}, 23(1987), 510-524.

\bibitem{RN27}
Kodama Y.,  Hasegawa A.,
\newblock Fission of Optical Solitons Induced by Stimulated Raman Effect.
\newblock {\em Optics Letters}, 13(1988), 392-394.

\bibitem{RN28}
Tsuru H.,   Wadati M.,
\newblock The Multiple Pole Solutions of the Sine-Gordon Equation.
\newblock {\em Journal of the Physical Society of Japan}, 53(1984), 2908-2921.

\bibitem{RN29}
Wadati M.,  Ohkuma K.,
\newblock Multiple-Pole Solutions of the Modified Korteweg-de Vries Equation.
\newblock {\em Journal of the Physical Society of Japan}, 51(1982), 2029-2035.

\bibitem{RN30}
Zhang YS.,  Tao XX.,  Xu SW.,
\newblock The Bound-State Soliton Solutions of the Complex Modified KdV Equation.
\newblock {\em Inverse Problems}, 36(2020).


\bibitem{zy2019}
Zhang Y.S., Rao J.G., Cheng Y. and He J.S.: Riemann-Hilbert Method for the Wadati-Konno-Ichikawa Equation: N Simple Poles and One Higher-Order Pole, Phys. D, \textbf{399}, 173-185 (2019).

\bibitem{zys2020}
Zhang Y.S., Tao X.X., Yao T.T. and He J.S.: The Regularity of the Multiple Higher-Order Poles Solitons of the NLS Equation, Stud. Appl. math., (September 2020).


\bibitem{RN1}
Zakharov   V.E.,   Shabat   A.B.,
\newblock Exact theory of two-dimensional self-focusing and one-dimensional self-modulation of waves in nonlinear media.
\newblock {\em  Sov. Phys. JETP}, 34(1972), 62-69.






\bibitem{RN2}
Demontis  F.,  Prinari  B., van der Mee  C., Vitale  F.,
\newblock The inverse scattering transform for the defocusing nonlinear schr\"{o}dinger equations with nonzero boundary conditions.
\newblock {\em Studies in Applied Mathematics}, 131(2013), 1-40.

\bibitem{RN3}
Kakei  S., Sasa  N., Satsuma  J.,
\newblock Bilinearization of a generialized derivative nonlinear Schr\"{o}dinger equation.
\newblock {\em Journal of The Physical Society of Japan},
  64(1995), 1519-1523.


\bibitem{RN5}
Kundu  A.,
\newblock Exact solutions to higher-order nonlinear equations through gauge transformation.
\newblock {\em Physica D}, 25(1987), 399-406.

\bibitem{EN}
Fan  E.G.,
\newblock A family of completely integrable multi-Hamiltonian systems explicitly related to some celebrated equations.
\newblock {\em Journal Mathematical Physics}, 42(2001), 4327-4344.


\bibitem{RN8}
Kaup  D.J., Newell  A.C.,
\newblock An exact solution for a derivative nonlinear Schr\"{o}dinger equation.
\newblock {\em Journal of Mathematical Physics}, 19(1978),  798-801.

\bibitem{RN6}
Chen  H.H., Lee  Y.C., Liu  C.S.,
\newblock Integrability of nonlinear Hamiltonian systems by inverse scattering method.
\newblock {\em Physica Scripta}, 20(1979), 490-492.

\bibitem{RN7}
Gerdjikov  V.S., Ivanov  I.,
\newblock A quadratic pencil of general type and nonlinear evolution equations.\uppercase\expandafter{\romannumeral2}.Hierarchies of Hamiltonian structures.
\newblock {\em Bulgarian Journal of Physics}, 10(1983), 130-143.

\bibitem{APPLI1}
Tzoar  N., Jain  M.,
\newblock {Self-phase modulation in long-geometry optical waveguides}.
\newblock{\em Physical Review A}, 23(1981), 1266-1270.

\bibitem{APPLI2}
Kodama  Y.,
\newblock{Optical solitons in a monomode fiber}.
\newblock{\em Journal of Statistical Physics}, 39(1985), 597-614.


\bibitem{APPLI4}
Rogister  A.,
\newblock{Parallel Propagation of Nonlinear Low-Frequency Waves in High-$\beta$ Plasma}.
\newblock{\em Physics of Fluids}, 14(1971), 2733-2739.

\bibitem{APPLI5}
Nakatsuka H., Grischkowsky  D., Balant A.C.,
\newblock{Nonlinear Picosecond-Pulse Propagation through Optical Fibers with Positive Group Velocity Dispersion}.
\newblock {\em Physical Review Letters}, 47(1981), 910-913.



\bibitem{APPLI6}
Einar M.,
\newblock{Nonlinear alfv$\acute{e}$n waves and the dnls equation: oblique aspects}.
\newblock {\em Physica Scripta}, 40(1989), 227-237.

\bibitem{APPLI7}
Agrawal  G.P.,
\newblock{Nonlinear fiber optics}.
\newblock Academic Press, Boston, 2007.




\bibitem{RN9}
Fan  E.G.,
\newblock Darboux transformation and solion-like solutions for the Gerdjikov-Ivanov equation.
\newblock {\em Journal of Physics A}, 33(2000), 6925-6933.



\bibitem{RN11}
Dai  H.H., Fan  E.G.,
\newblock Variable separation and algebro-geometric solutions of the Gerdjikov-Ivanov equation.
\newblock {\em Chaos, Solitons and Fractals}, 22(2004), 93-101.

\bibitem{RN12}
Hou  Y., Fan  E.G., Zhao  P.,
\newblock Algebro-geometric solutions for the Gerdjikov-Ivanov hierarchy.
\newblock {\em Journal of Mathematical Physics}, 54(2013), 1-30.

\bibitem{RN13}
He  B., Meng  Q.,
\newblock Bifurcations and new exact travelling wave solutions for the Gerdjikov-Ivanov equation.
\newblock {\em Communications in Nonlinear Science and Numerical Simulation}, 15(2010), 1783-1790.

\bibitem{RN14}
Kakei  S., Kikuchi  T.,
\newblock Solutions of a derivative nonlinear Schr\"{o}dinger hierarchy and its similarity reduction.
\newblock {\em Glasgow Mathematical Journal}, 47(2005), 99-107.

\bibitem{NZG}
 Nie   H.,  Zhu  J.Y.,   Geng  X.G.,
\newblock {Trace formula and new form of N-soliton to the Gerdjikov-Ivanov equation}.
\newblock {\em Anal. Math. Phys.},  8(2018),  415-426.



\bibitem{RN31}
Zhang ZC.,  Fan EG.,
\newblock Inverse scattering transform for the Gerdjikov-Ivanov equation with nonzero boundary conditions.
\newblock {\em Z.  Angew.   Math.  Phys.}, 71(2020).




\end{thebibliography}
\end{document}